%% file: TSPonPLG.tex
\documentclass[a4paper,12pt]{article}

\input{mypreambel}

\begin{document}

\title{\bf Approximability of TSP on Power Law Graphs\\[1ex]}
\author{Mikael Gast\thanks{Dept. of Computer Science and the Hausdorff Center for Mathematics, University of Bonn.
    e-mail:{ \texttt{\href{mailto:gast@cs.uni-bonn.de}{gast@cs.uni-bonn.de}}}} \and
	Mathias Hauptmann\thanks{Dept. of Computer Science, University of Bonn.
    e-mail:{ \texttt{\href{mailto:hauptman@cs.uni-bonn.de}{hauptman@cs.uni-bonn.de}}}} \and
	Marek Karpinski\thanks{Dept. of Computer Science and the Hausdorff Center for Mathematics, University of Bonn. Research supported by DFG grants and the Hausdorff grant EXC59-1.
    e-mail:{ \texttt{\href{mailto:marek@cs.uni-bonn.de}{marek@cs.uni-bonn.de}}}}}
\date{}
\maketitle

\begin{abstract}
 In this paper we study the special case of Graphic TSP where the underlying graph is a \emph{power law graph} (PLG). We give a refined analysis of some of the current best approximation algorithms and show that an improved approximation ratio can be achieved for certain ranges of the \emph{power law exponent} $\beta$.
For the value of power law exponent $\beta=1.5$ we obtain an approximation ratio of $1.34$ for Graphic TSP. 
Moreover we study the $(1,2)$-TSP with the underlying graph of $1$-edges being a PLG. We show improved approximation ratios in the case of underlying deterministic PLGs for $\beta$ greater than $1.666$. 
For underlying \emph{random} PLGs we further improve the analysis and show even better expected approximation ratio for the range of $\beta$ between $1$ and $3.5$. 
On the other hand we prove the first explicit inapproximability bounds for $(1,2)$-TSP for an underlying power law graph.
\end{abstract}


\input{Intro}
\input{Approx}
\input{Lower}

\newpage
\section{Further Research}
It remains an interesting open problem on improving our approximability results as well as their extension to other related cases.
It would be very interesting to shed some more light on the structural properties of the underlying problems.

\printbibliography


%

\end{document}

%% file: mypreambel.tex
\usepackage[utf8]{inputenc}

\usepackage{geometry}

\usepackage{pgf}
\usepackage{nicefrac}
\usepackage{booktabs}
\usepackage{fullpage}
\usepackage[]{asymptote}

\usepackage[%
]{graphicx}

\usepackage[
   centertags, 
   sumlimits,  
   intlimits,  
   namelimits, 
   fleqn,     
]{amsmath} %

\usepackage{amsthm}
\usepackage{amssymb}
\usepackage{mathtools}

\usepackage[T1]{fontenc}
\usepackage{lmodern}
\usepackage[%
	english
]{nomencl}

\definecolor{sectioncolor}{RGB}{0, 0, 0}    
\definecolor{textcolor}{RGB}{0, 0, 0}        
\definecolor{shadecolor}{gray}{0.90}
\definecolor{pdfurlcolor}{rgb}{0,0,0.6}
\definecolor{pdffilecolor}{rgb}{0.7,0,0}
\definecolor{pdflinkcolor}{rgb}{0,0,0.6}
\definecolor{pdfcitecolor}{rgb}{0,0,0.6}
\colorlet{stringcolor}{green!40!black!100}
\colorlet{commencolor}{blue!0!black!100}

\usepackage{aliascnt}

\usepackage[natbib=true,
	    bibencoding=utf8,
	    sorting=nyt,
	    style=alphabetic,
	    useprefix=true,
            hyperref=true,
            backref=false,
	    backrefstyle=none,
	    alldates=short,
            abbreviate=true,
	    isbn=false,
	    doi=false,
	    firstinits=true,
	    maxbibnames=99,
	    backend=biber]{biblatex}
	    
\DeclareBibliographyCategory{needsurl}
\newcommand{\entryneedsurl}[1]{\addtocategory{needsurl}{#1}}
\renewbibmacro*{url+urldate}{%
  \ifcategory{needsurl}{
    \printfield{url}%
    \iffieldundef{urlyear}
      {}
      {\setunit*{\addspace}%
       \printurldate}}
    {}}

\renewbibmacro{in:}{}

\theoremstyle{plain}
 
\newtheorem{theorem}{Theorem}  
 
\newaliascnt{lemma}{theorem}  
\newtheorem{lemma}[lemma]{Lemma}  
\aliascntresetthe{lemma}

\newtheorem*{theorem*}{Theorem}

\newtheorem*{corollary}{Corollary}

\theoremstyle{definition}
\newtheorem{definition}{Definition}[section]

\theoremstyle{remark}

\usepackage[
  colorlinks=true,         
  urlcolor=pdfurlcolor,    
  filecolor=pdffilecolor,  
  linkcolor=pdflinkcolor,  
  citecolor=pdfcitecolor,  %
]{hyperref}

\usepackage{csquotes}
\usepackage[
	english,
]{babel}
\DeclareMathOperator{\e}{e}

\newcommand{\E}{\mathbb{E}}

\renewcommand{\deg}[1]{\ensuremath{d(#1)}}
\newcommand{\gab}{\ensuremath{G_{\alpha,\beta}}}
\newcommand{\Gab}{\ensuremath{\mathcal{G}_{\alpha,\beta}}}
\newcommand{\Mab}{\ensuremath{P(\alpha,\beta)}}

\renewcommand{\leq}{\leqslant}
\renewcommand{\geq}{\geqslant}

%% file: Intro.tex
\section{Introduction}

Given a set of cities along with the distance between each pair of them, the \emph{Traveling Salesman Problem} (TSP) is to find the shortest tour through all the cities and returning to the starting point.
The TSP is one of the best known problems in combinatorial optimization and has many applications, for example in the fields of logistics, genetics and telecommunications.
It is well known that the general TSP is $NP$-hard \cite{Karp1972} and that, unless $P=NP$, there can be no efficient approximation algorithm for the TSP that achieves some bounded approximation ratio \cite{Sahni1976}.
For the case of the (symmetric) \emph{Metric TSP} with a distance function satisfying the triangle inequality Christofides devised a $\nicefrac{3}{2}$-approximation algorithm \cite{Christofides1976} which has seen no improvement for nearly four decades.
But a number of improvements were made for some special cases of Metric TSP (see e.g. \cite{Berman2006,Sebo2014}). 
We refer to 
\entryneedsurl{Taxonomy2015}\cite{Taxonomy2015} and \cite{Karpinski2015} for an overview of some of the current approximability and inapproximability results.

In this paper we study two special cases of Metric TSP, namely, the \emph{$(1,2)$-TSP} with 
distances between cities being either one or two, and the \emph{Graphic TSP} where we assume that the metric is the shortest path metric of an undirected graph.
More precisely we are interested in the case of Graphic TSP where the underlying graph is a \emph{power law graph}, and, for the $(1,2)$-TSP, where the undirected graph of $1$-edges is a power law graph.

Power law graphs have the property that their node degree distribution follows a power law, 
i.e. the number of nodes of degree $i$ is proportional to $i^{-\beta}$, for some fixed power law exponent $\beta >0$.
This parameter $\beta$ is called the power law exponent.
Many of the existing real-world networks turn out to be power law graphs. This has been observed for the graphs of the Internet and the World Wide Web (WWW), 
peer-to-peer networks, gene regulatory networks and protein interaction networks and for social networks.

\paragraph{Our Results}
In this paper we study the special case of the Graphic TSP when the underlying graph is a power law graph.
We denote the special case as Power Law Graphic TSP. 
For the case when the power law exponent $\beta$ satisfies 
$1<\beta <2.48$ and the graph is connected,
we show that the problem is approximable within $\frac{1}{2}+\frac{\zeta (\beta)}{\max\{2,\zeta (\beta)+1\slash 2\}}$, where $\zeta (\beta)=\sum_{i=1}^{\infty}i^{-\beta}$
is the \emph{Riemann Zeta Function}. 
For $\beta>1.256$, this improves over the current best Graphic TSP upper approximation bound of $\nicefrac{7}{5}$ due to \cite{Sebo2014}.
In particular, for $\beta =1.5$ this yields an approximation ratio of $1.339$. For $\beta =2.5$, the resulting approximation ratio is $1.171$. For $\beta >2.479$, power law graphs are not connected any more.  

Then we consider the $(1,2)$-TSP with underlying power law graph, which we denote as Power Law $(1,2)$-TSP. 
For the case $\beta >1$, we obtain an approximation ratio of 
$\frac{\frac{11}{9}\zeta (\beta)+\frac{29}{72}}{\zeta (\beta )+\frac{1}{2}}$.
For $\beta>2.729$ we obtain an improved ratio of $\frac{2\zeta(\beta) + \frac{1}{2}\zeta(\beta-1) -1}{\zeta(\beta)+ \frac{1}{2}}$ which converges to $1$ as $\beta\to\infty$.
For $\beta>1.666$, this improves over the current best $(1,2)$-TSP upper approximation bound of $\nicefrac{8}{7}$ due to \cite{Berman2006}.
For $\beta =1.5$, this yields a $1.155$-approximation and for $\beta =2.5$ a 
$1.109$-approximation. 

In the case when the underlying graph is a random PLG with $\beta>1$, we obtain an improved expected approximation ratio. 
In the case $1<\beta<2$ and $\beta=2$ this is equal to $\frac{\frac{11}{9}\zeta(\beta)+\frac{29}{36}\cdot\frac{5}{4}}{\zeta(\beta)+\frac{5}{4}}$ and $\frac{\frac{11}{9}\zeta(\beta)+\frac{29}{36}\cdot\frac{5}{8}}{\zeta(\beta)+\frac{5}{8}}$, respectively.
For $\beta>2$ the expected approximation ratio is $\frac{\frac{11}{9}\zeta(\beta)+\frac{29}{36}\cdot E_k}{\zeta(\beta)+E_k}$, where $E_k$ is the expected lower bound on the number of $2$-edges in an optimum tour. 

On the other hand we show that for $\beta >1$, the Power Law $(1,2)$-TSP is NP-hard to approximate within approximation ratio $\frac{(\zeta (\beta)+\nicefrac{1}{2})\cdot 3^{\beta -1}\cdot 2\cdot (\beta -1)\cdot 354 +1}{(\zeta (\beta)+\nicefrac{1}{2})\cdot 3^{\beta -1}\cdot 2\cdot (\beta -1)\cdot 354}$.
This gives an approximation lower bound of $1.00086$ for $\beta =1.2$ and of $1.0012$ for $\beta =1.1$.


\paragraph{Organization of the Paper} In \autoref{sec:Methods} we give an outline of the methods and constructions used in the paper. \autoref{sec:Preliminaries} provides the definition of the PLG model due to \cite{Aiello2001} and related notations.
In \autoref{sec:GraphicTSP} we describe our results for the Graphic TSP on power law graphs.
In \autoref{sec:OneTwoTSP} we present our results on the $(1,2)$-TSP with underlying power law graph. \autoref{sec:DeterministicUpper} contains the analysis of the algorithm from \cite{Papadimitriou1993} for deterministic power law graphs. \autoref{sec:AsymptoticUpper} deals with the case of a large power law exponent. In \autoref{sec:RandomUpper} we consider the case when the underlying graph is a random PLG. Approximation lower bounds are given in \autoref{sec:LowerBounds}.

\section{The Method}\label{sec:Methods}

We give first a brief outline of the methods and constructions used in this paper.
For the Graphic TSP with underlying power law graph, we give an improved analysis of the algorithms of Christofides \cite{Christofides1976}, M{\"o}mke-Svensson \cite{Moemke2011} and the algorithm of Mucha \cite{Mucha2012}.
This is based on the analysis of local configurations in PLGs, consisting of a high-degree node and its degree $1$ and $2$ neighbors. We show that any such configuration
contributes to the lower bound for the length of an optimum tour. The same idea also applies to the $(1,2)$-TSP with underlying PLG. 

The approximation hardness results for the $(1,2)$-TSP with underlying PLG are based on efficient embeddings of graphs into power law graphs. 
We start from the lower bound for $(1,2)$-TSP given in \cite{Karpinski2013}. We map the associated graphs to instances which contain a perfect matching. This is then used in order
to embed those graphs into sufficiently small PLGs. 

\section{Preliminaries}\label{sec:Preliminaries}

In this section we first give the formal definition of $(\alpha,\beta)$-power law graphs. Then we provide estimates for 
sizes and volumes of some subsets of the vertex set of a power law graph which we call \emph{intervals}. Later on we make use of 
these estimates in the analysis of our upper and lower bound constructions.
\begin{definition}\cite{Aiello2001}
An undirected multigraph $G=(V,E)$ with self loops is called an $(\alpha,\beta)$ power law graph if the following conditions hold:
\begin{itemize}
\item The maximum degree is $\Delta=\lfloor e^{\alpha\slash\beta}\rfloor$.
\item For $i=1,\ldots , \Delta$, the number $y_i$ of nodes of degree $i$ in $G$ satisfies
      \[y_i = \left\lfloor \frac{e^{\alpha}}{i^{\beta}}\right\rfloor\] 
\end{itemize}
\end{definition}
The following estimates for the number $n$ of vertices of an $(\alpha,\beta)$ power law graph are well known \cite{Aiello2001}:
\[n\approx \left\{\begin{array}{l@{\quad}l}
 \frac{e^{\alpha\slash\beta}}{1-\beta} & \mbox{for $0<\beta <1$,}\\
 \alpha\cdot e^{\alpha} & \mbox{for $\beta =1$,}\\
 \zeta (\beta )\cdot e^{\alpha} & \mbox{for $\beta >1$.}
\end{array}\right.\quad 
m\approx \left\{\begin{array}{l@{\quad}l}
 \frac{1}{2}\frac{e^{2\alpha\slash\beta}}{2-\beta} & \mbox{for $0<\beta <2$,}\\
 \frac{1}{4}\alpha e^{\alpha} & \mbox{for $\beta =2$,}\\
 \frac{1}{2}\zeta (\beta -1)e^{\alpha} & \mbox{for $\beta >2$.}
\end{array}\right.\]
Here $\zeta (\beta)=\sum_{i=1}^{\infty}i^{-\beta}$ is the {\sl Riemann Zeta Function}.

By \Mab\ we denote the probability distribution on the set \Gab\ which is obtained in the following way \cite{Aiello2001}:
\begin{enumerate}
  \item Generate a set $L$ of $\deg{v}$ distinct copies of each vertex $v$.
  \item Generate a random matching on the elements of $L$.
  \item For each pair of vertices $u$ and $v$, the number of edges joining $u$ and $v$ in $\gab$ is equal to the number of edges in the matching of $L$, which join copies of $u$ to copies of $v$.
\end{enumerate}

\begin{figure}[htb]
\centering
 \includegraphics[scale=1]{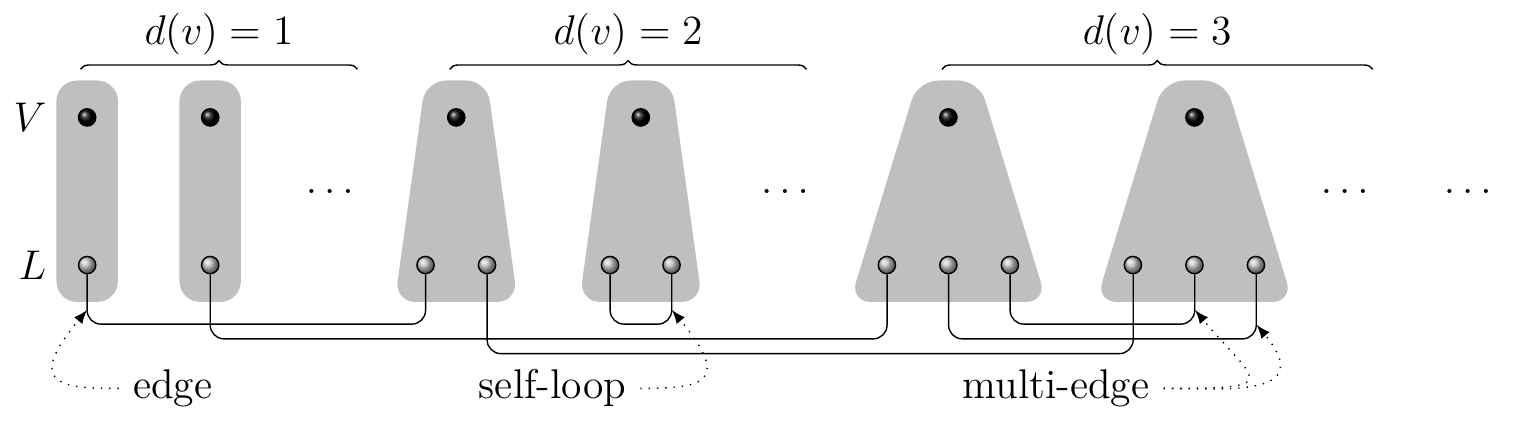}
 \caption{}
 \label{fig:RandomMatch}
\end{figure}

As in \cite{Aiello2001}, in the following we work with the real numbers $\frac{\e^{\alpha}}{i^{\beta}}$, $\e^{\frac{\alpha}{\beta}}$ instead of their integer counterparts. For $\beta > 2$ the overall error is bounded by $\e^{\frac{\alpha}{\beta}}=o(n)$ (c.f. \cite{Aiello2001}, remark on page 6).

 Given an $(\alpha,\beta)$ power law graph $G=(V,E)$ with $n$ vertices and maximum degree $\Delta$ and two integers $1\leq a\leq b\leq\Delta$, 
 an \emph{interval} $[a,b]$ is defined as the subset of $V$
 \begin{equation*}
 [a,b] = \{v\in V|a\leq \mbox{deg}_G(v)\leq b\}
 \end{equation*}
 If $U\subseteq V$ is a subset of vertices, the \emph{volume} $\mbox{vol}(U)$ of $U$ is defined as the sum of node degrees of nodes in $U$.
 In \autoref{sec:GraphicTSP} we will make use of estimates of sizes and volumes of node intervals in $(\alpha,\beta)$-PLGs.
 In particular we will show that in the case $\beta >2$ the volume of a node degree interval of the form 
$[x\Delta,\Delta ]=\{v|x\Delta\leq\mbox{deg}(v)\leq\Delta\}$ is bounded by $ \frac{e^{2\alpha\slash\beta}}{(\beta -2)x^{\beta -2}}$.

%% file: Approx.tex
\section{Graphic TSP on Power Law Graphs}\label{sec:GraphicTSP}
In this section we consider the Graphic TSP in power law graphs. Given a connected graph $G=(V,E)$ with $n$ vertices, the Graphic TSP asks for a minimum cost tour 
on the vertices with respect to the shortest path metric induced by $G$.
This is a natural special case of the general TSP and it was shown that an improvement of the approximation algorithm of Christofides' algorithm can be achieved.
The first small improvement is due to \citeauthor{Gharan2011} \cite{Gharan2011} who showed a $(\frac{3}{2}-\varepsilon)$-approximation (for an $\varepsilon$ in the order of $10^{-12}$) based on a refined analysis of Christofides' algorithm for Graphic instances.
A whole different approach is due to \citeauthor {Moemke2011} \cite{Moemke2011} where the starting point is a $2$-vertex connected graph which is made Eulerian by removing edges. This approach yields an approximation ratio of $\frac{14(\sqrt{2}-1)}{12\sqrt{2}-13}<1.461$ which was subsequently improved by \citeauthor{Mucha2012} \cite{Mucha2012} to a ratio of $\frac{13}{9}\approx 1.444$.
The current best approximation ratio of $\frac{7}{5}$ is achieved by \citeauthor{Sebo2014} \cite{Sebo2014}.

Here we consider the special case of Graphic TSP when the underlying graph $G$ is an $(\alpha,\beta )$-PLG.
We show that---for certain ranges of the parameter $\beta$---a refined analysis of the above algorithms yield better approximation ratios in this case.
Our results hold for $1<\beta<2.479$, namely when $\zeta(\beta-1)>2\cdot\zeta(\beta)$.
For $\zeta(\beta-1)\leq 2\cdot\zeta(\beta)$, $(\alpha,\beta)$-PLGs are not connected anymore.

\subsection{Improved Analysis}
We start by slightly improving the analysis of the basic $2$-approximation algorithm for TSP applied on instances of Graphic TSP with an underlying $(\alpha,\beta )$-PLG.
For the general TSP instances, a $2$-approximation can easily be achieved by the so called \emph{MST heuristic}. In order to construct a tour one computes a 
\emph{minimum spanning tree} (MST) of $G$ and replaces each edge of the tree by a pair of parallel edges. The resulting multigraph is Eulerian, and making use of triangle inequality,
an Euler tour can be transformed into a traveling salesman tour. 
Since any tour is connected and thus contains a spanning tree, a MST has cost at most half of the optimum tour. 
For the case of an underlying $(\alpha,\beta )$-PLG the cost of an MST is $\zeta (\beta)e^{\alpha}-1$, while every tour has cost at least 
$\zeta (\beta)e^{\alpha}+\frac{e^{\alpha}}{2}$. This lower bound holds since $G$ contains $e^{\alpha}$ nodes of degree $1$. 
As we shall see below in the analysis of Christofides' algorithm on PLGs, we also have $c(\tau )\geq 2e^{\alpha}$ for every tour $\tau$.
Thus we obtain the following result.
\begin{lemma}
For $1<\beta<2.48$, the approximation ratio of the MST heuristic for $(\alpha,\beta)$-Power Law Graphic TSP is bounded by 
$\frac{2\zeta (\beta)}{\max\{2,\zeta (\beta)+\frac{1}{2}\}}$.
\end{lemma}  
Now we consider the performance of Christofides' algorithm \cite{Christofides1976} on power law Graphic instances. On general instances, 
this algorithm constructs an MST $T$ and a minimum cost perfect matching $M$ on the subset
of nodes which have an odd degree in $T$. If $\tau^*$ denotes an optimum TSP tour, then $c(M)\leq c(\tau^*)$ and $c(M)\leq c(\tau^*)\slash 2$, which yields the 
approximation ratio of $1.5$ for the general Metric TSP. Now we will improve this bound for the special case of $(\alpha,\beta)$-Power Law Graphic TSP when $\beta >1$. 
Again the MST cost is $c(T)=\zeta (\beta)e^{\alpha}-1\leq c(\tau^*)$. The problem now is that we do not know in advance if the matching cost $c(M)$ is close to $c(T)$
or close to $\frac{1}{2}c(T)$. 

Our approach is to charge at least some part of the matching cost in order to get an improved approximation ratio. First we observe that all degree $1$ nodes in $G$ 
have an odd degree in $T$. Thus we obtain that 
\[c(M) = e^{\alpha} +(c(M)-e^{\alpha}) \leq e^{\alpha} +\left (\frac{c(\tau^*)}{2}-e^{\alpha}\right ),\]
from which we also conclude that $c(\tau^* )\geq 2e^{\alpha }$. Moreover, $c(\tau^*)\geq (\zeta (\beta )+\frac{1}{2})e^{\alpha}$.
Hence the approximation ratio is bounded by
\[\frac{\zeta (\beta)e^{\alpha}+c(M)}{c(\tau^*)}\leq\frac{\zeta (\beta)e^{\alpha}+\frac{c(\tau^*)}{2}}{c(\tau^*)}
\leq\frac{1}{2}+\frac{\zeta (\beta)}{\max\{2,\zeta (\beta)+\frac{1}{2}\}}\]
We have shown the following result.
\begin{lemma}
For $1<\beta<2.48$, the approximation ratio of Christofides' algorithm on $(\alpha,\beta)$-power law Graphic instances is at most 
$\frac{1}{2}+\frac{\zeta (\beta)}{\max\{2,\zeta (\beta)+\frac{1}{2}\}}$.
\end{lemma}
Now we consider the M{\"o}mke-Svensson algorithm for TSP \cite{Moemke2011}. It is based on the notion of {\sl removable pairings}. A {\sl removable pairing} $(R,{\mathcal P})$
consists of a subset of edges $R\subseteq E$ and a subset ${\mathcal P}\subseteq P(E)$ of the power set of $E$ such that ${\mathcal P}=\{P_v|v\in U\}$ for some subset
$U\subseteq V$, each $v\in U$ has degree at least $3$ and $P_v$ is a set of two edges incident to $v$. The sets in ${\mathcal P}$ have to be pairwise disjoint. Moreover,
for each subset $F\subseteq R$ such that for each $v\in U$, $F$ contains at most one element from $P_v$, the graph $(V,E\setminus F)$ is connected.

\citeauthor{Moemke2011} \cite{Moemke2011} show that, given a removable pairing $(R,{\mathcal P})$, one can construct in polynomial time a tour $\tau$ of cost 
$c(\tau )\leq \frac{4}{3}c(E)-\frac{2}{3}c(R)$. This yields an approximation ratio of $\frac{4}{3}$ for Subcubic Graphic TSP, i.e. Graphic TSP in graphs of maximum degree
at most $3$. We now analyse the performance of this algorithm for Power Law Graphic TSP for $\beta >2$. Note that for $\beta\leq 2$ the M{\"om}ke-Svensson bound 
is not applicable, since the cardinality of $E$ is not linear in $|V|$ anymore. In the case $\beta >2$, we use our lower bound
$c(\tau^*)\geq\max\{2,\zeta (\beta)+\frac{1}{2}\}\cdot e^{\alpha}$. We apply the M{\"o}mke-Svensson algorithm to every $2$-vertex connected component of the
graph $G=(V,E)$. Suppose first that the subgraph $G'=(V',E')$ induced by the vertices of degree at least $2$ is $2$-vertex connected.
Now we have to give a bound on the size of the set $R$, where $(R,{\mathcal P})$ is a removable pairing
used in the algorithm. This is constructed as follows. Choose a root $v\in V$ and let $(V,S)$ be a DFS tree in $G'$ rooted at $r$. Initially, $R\colon =E\setminus S$
and ${\mathcal P}=\emptyset$. Then 
the vertices $v\in V$ are processed in any order. For each such $v$, if $v$ has degree at least $3$, is incident to some edge $e=\{v,w\}\in E\setminus S$ such that 
$v$ is on the $r-w-$path in $(V,S)$, $e'$ is the first edge on the subpath from $v$ to $w$ and $e'$ has not yet been added to $R$, then add $e'$ to $R$ and $\{e,e'\}$ to 
the set ${\mathcal P}$.

First we observe that in the case $\beta >2$, 
\begin{align*}
|E'\setminus S| & \geq  |E|-|\{v|\text{deg}(v)=1\}|-(|V'|-1)\\
                & \geq  \frac{1}{2}\zeta (\beta -1)e^{\alpha} - e^{\alpha}- (\zeta (\beta)-1)e^{\alpha}
                 =      \left(\frac{1}{2}\zeta (\beta -1)-\zeta (\beta)\right)e^{\alpha} 
\end{align*}  
Given an edge $e\in E'\setminus S$ with $e=\{v,w\}$ such that $v$ lies on the $r-w-$path in the DFS tree, then $v$ is called the root of the edge $e$.
Since each such root gives rise to at least one additional edge in the set $R$, we will now give a lower bound on the number of roots needed to cover all
the edges in $E'\setminus S$. Each root $v\in V'\setminus\{r\}$ has degree at least $2$ plus the number of edges $e\in E'\setminus S$ with root $v$.

The volume of a set of nodes is defined as the sum of their node degrees. In Power Law Graphs, in the case $\beta >2$ the volume of a node degree interval
$[x\Delta,\Delta ]=\{v|x\Delta\leq\text{deg}(v)\leq\Delta\}$ can be estimated as follows:
\begin{align*}
\text{vol}([x\Delta,\Delta ]) & = \sum_{i=x\Delta}^{\Delta}\left\lfloor\frac{e^{\alpha}}{i^{\beta}}\right\rfloor\cdot i
    \leq \sum_{i=x\delta}^{\Delta}\frac{e^{\alpha}}{i^{\beta-1}}
    \leq e^{\alpha}\int_{x\Delta}^{\Delta +1}\frac{1}{x^{\beta -1}}dx + (1-x)\Delta\\ &= (1+o(1))e^{\alpha} \left [\frac{y^{2-\beta}}{2-\beta}\right ]_{x\Delta}^{\Delta +1}
    = \frac{(1+o(1))\e^{\alpha}}{\beta -2}\left (\frac{1}{(x\Delta)^{\beta -2}}-\frac{1}{(\Delta +1)^{\beta -2}}\right ) \\
     &\leq   \frac{e^{\alpha-\alpha\cdot\frac{\beta-2}{\beta}}}{(\beta -2)\cdot x^{\beta -2}}
     =      \frac{e^{2\alpha\slash\beta}}{(\beta -2)x^{\beta -2}}
\end{align*} 	
The volume available for edges $e\in E'\setminus S$ is 
\[\sum_{i=x\Delta}^{\Delta}\left\lfloor\frac{e^{\alpha}}{i^{\beta}}\right\rfloor (i-2)\leq \frac{e^{2\alpha\slash\beta}}{(\beta -2)x^{\beta -2}}-2(1-x)\Delta\]
Now this volume has to cover the set $E'\setminus S$, i.e.
\begin{align*}
   \frac{e^{2\alpha\slash\beta}}{(\beta -2)x^{\beta -2}}-2(1-x)\Delta & \geq \left(\frac{1}{2}\zeta (\beta -1)-\zeta (\beta)\right)e^{\alpha}\shortintertext{which holds if and only if}
   \frac{e^{2\alpha\slash\beta}}{(\beta -2)x^{\beta -2}} & \geq (1+o(1))\cdot \left(\frac{1}{2}\zeta (\beta -1)-\zeta (\beta)\right)e^{\alpha}
 \end{align*}
Hence we have to choose 
\[x\leq\frac{e^{-\alpha\slash\beta}}{((\beta -2)(\frac{1}{2}\zeta (\beta -1)-\zeta (\beta)))^{1\slash (\beta -2)}} \]
Then the size of the set $[x\Delta,\Delta ]$ is
\begin{align*}
\sum_{x\Delta}^{\Delta}\left\lfloor\frac{e^{\alpha}}{i^{\beta}}\right\rfloor & \geq e^{\alpha}\cdot\frac{(x\Delta)^{1-\beta}-(\Delta +1)^{1-\beta}}{\beta -1}-\Delta\\
  & \geq (1-o(1))\frac{e^{\alpha}}{(\beta -1)}\cdot \left [(\beta -2)\left(\frac{1}{2}\zeta (\beta -1)-\zeta (\beta)\right) \right ]^{\frac{\beta -1}{\beta -2}}
\end{align*}
Thus we obtain an improved lower bound on the size of the set $R$ as follows:
\begin{align*}
|R| & \geq |E\setminus S|\\
    & \quad +(1-o(1))e^{\alpha}\frac{[(\beta -2)(\zeta (\beta -1)\slash 2-\zeta (\beta)+1 ]^{\frac{\beta -1}{\beta -2}}}{\beta -1}
\end{align*}
We let $t_{\beta}:= e^{\alpha}\frac{[(\beta -2)(\zeta (\beta -1)\slash 2-\zeta (\beta)+1 ]^{\frac{\beta -1}{\beta -2}}}{\beta -1}$.
We obtain that the cost of the tour $\tau$ constructed by the M{\"o}mke-Svensson algorithm is bounded as
\begin{align*}
c(\tau ) & \leq \frac{4}{3}|E| - \frac{2}{3}|R|\\
         & \leq \frac{2}{3}|E|+\frac{2}{3}|S|-\frac{2}{3}t_{\beta}e^{\alpha}\\
         & \leq \left (\frac{2}{3}\zeta (\beta -1)+\frac{2}{3}\zeta (\beta)+\frac{5}{6}\right )e^{\alpha}
\end{align*}
Now we consider the general case when the subgraph induced by the vertices of degree at least $2$ is not $2$-vertex connected.
Let $G=(V,E)$ be the power law graph and $G'=(V',E')$ the subgraph induced by $\{v\in V|\text{deg}_G(v)\geq 2\}$. Suppose  that
$G'$ consists of $2$-vertex connected components $G'_1,\ldots , G'_L$, which are connected in a tree-like manner by bridges and/or articulation points.
Then we run the M{\"o}mke-Svensson algorithm separately on every component $G'_i=(V'_i,E_i),1\leq i\leq L$. Suppose that $\tau_1,\ldots , \tau_L$ are the tours constructed
in this way, and $(R_i,{\mathcal P}_i),i=1,\ldots , L$ are the associated removable pairs. Thus we have $c(\tau_i)\leq \frac{4}{3}|E_i|- \frac{2}{3}|R_i|$.
Now we observe that $\sum_i|R_i|\leq\sum_i|S_i|-\frac{2}{3}t_{\beta}e^{\alpha}$.
Combined with our lower bound for the cost of a tour, this yields the following result.
\begin{lemma}
For $2<\beta <2.48$, the approximation ratio of the M{\"o}mke-Svensson algorithm on power-law Graphic instances is at most
$\frac{\frac{2}{3}\zeta (\beta -1)+\frac{2}{3}\zeta (\beta)+\frac{5}{6}}{\frac{1}{2}+\max\{2,\zeta (\beta)+\frac{1}{2}\}}$.
\end{lemma}
Finally we give an estimate for the performance of Mucha's algorithm \cite{Mucha2012} on power law graphic instances. 
He shows that a tour can be constructed whose cost is at most $\frac{1}{3}n+\frac{10}{9}x(E)$. Here $x(E)$ denotes an optimum solution 
for the Held-Karp relaxation.
Using this upper bound and the lower bound $\max\{2,\zeta (\beta)+\frac{1}{2}\}$, we obtain an A.R. of
$\frac{10}{9}+\frac{\frac{1}{3}\zeta (\beta)}{\max\{2,\zeta (\beta)+\frac{1}{2}\}}$.
In Figure \ref{fig:PlotGraphTSP} we show the graphs of the approximation ratio for the MST heuristic, Christofides' algorithm, the M{\"o}mke-Svensson algorithm
and Mucha's algorithm for power law Graphic instances. 

\begin{figure}[htb]
  \centering
  \input{PlotGraphTSP}
  \caption{Approximation ratios of the MST-algorithm (MST), \citeauthor{Christofides1976}' algorithm (Chr), \citeauthor{Mucha2012}'s algorithm (Mu) and the M\"{o}mke-Svensson algorithm (M-S) for Power Law Graphic TSP for $1<\beta<2.48$.}
  \label{fig:PlotGraphTSP}
\end{figure}
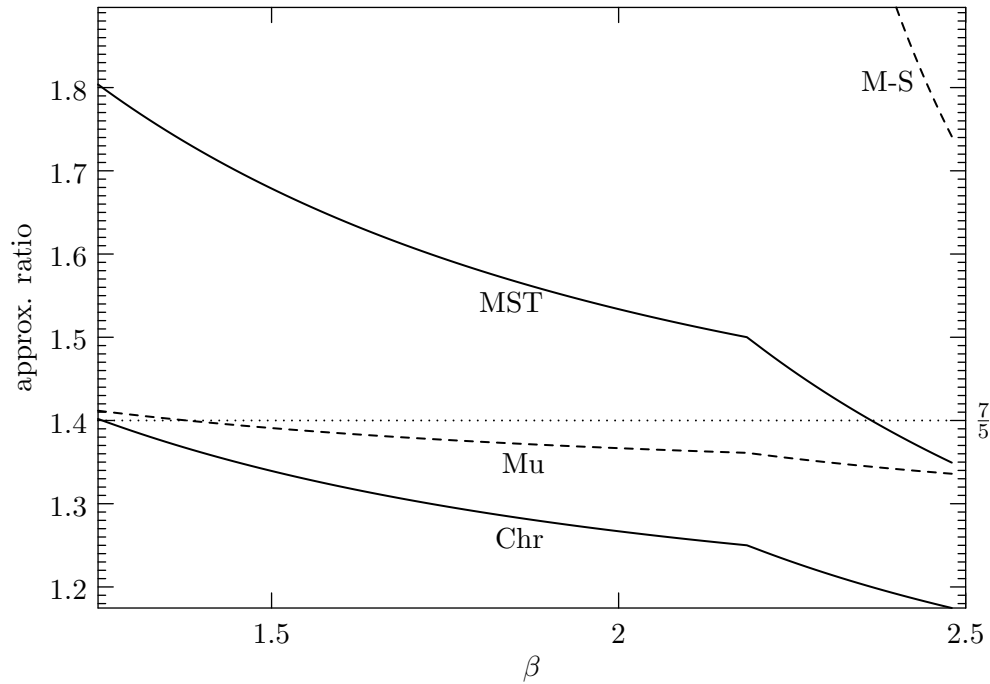

\section{(1,2)-TSP on Power Law Graphs}\label{sec:OneTwoTSP}

We consider now an interesting further special case of Metric TSP, namely the $(1,2)$-TSP where the distances between cities are either one or two. More precisely, we are given a complete undirected graph $G=(V,E)$ together with an edge weight function $w:E\to\{1,2\}$.
The problem is known to be NP-hard and MAX-SNP-complete as shown by \citeauthor{Papadimitriou1993} in \cite{Papadimitriou1993}. In the same paper they provided an approximation algorithm with approximation ratio $\frac{7}{6}$. After more than a decade this upper bound was slightly improved to $\frac{65}{56}$ by \citeauthor{Blaeser2005} in \cite{Blaeser2005}. The current best approximation algorithm is due to \citeauthor{Berman2006} and achieves an approximation ratio of $\frac{8}{7}$ as stated in \cite{Berman2006}.

Now we consider the special case of the $(1,2)$-TSP problem when the subgraph of $1$-edges is an $(\alpha,\beta )$-power law graph. In the following we refer to the subgraph of $1$-edges as the \emph{underlying} graph.

\subsection{Upper Bounds for Deterministic PLGs}\label{sec:DeterministicUpper}

Let $G=(V,E)$ be the underlying graph. This means that for $u\neq v$, the distance between $u$ and $v$ is $1$ iff $\{u,v\}\in E$, otherwise 
the distance is $2$. Now assume that $G$ is an $(\alpha,\beta )$-PLG for some $\beta >1$. 
In the following we make use of the methods and notations from \cite{Papadimitriou1993}. 

A \emph{cycle cover} $\mathcal{C}$ of a graph $G=(V,E)$ is a collection of vertex-disjoint cycles such that each vertex $v\in V$ is contained in a cycle of $\mathcal{C}$.
Suppose that there exists a cycle cover ${\mathcal C}$ (also called a 2-matching) of cost $n+k$, where $k$ is the number of $2$-edges, 
such that $\mathcal{C}$ does not contain any cycles of length less than four.  An optimum cycle cover with cycles of length at least four can be constructed
efficiently \cite{Hartvigsen1984,Manthey2008}.

Now let $v_1,v_2,\ldots , v_n,v_1$ be an optimum tour. Let $U$ be the set of nodes $v_i$ such that 
$d(v_i,v_{i+1})=2$. We may assume that ${\mathcal C}$ contains at most one cycle which also contains $2$-edges. 
All the other cycles of ${\mathcal C}$ are called pure.
Let $c_2$ be the number of pure cycles in ${\mathcal C}$ which contain at least one node from $U$.
In \cite{Papadimitriou1993}, two auxiliary graphs are constructed. 
First, the bipartite graph $B$ contains one node for each cycle from ${\mathcal C}$ and one node 
for each vertex from $V$. A cycle node - here just denoted as $C$, where $C\in {\mathcal C}$ - is connected to 
a node corresponding to vertex $v\in V$ if $v$ is not contained in $C$ but there exists an edge in $G$ connecting $v$ to a node in $C$.
Let $M$ be a maximum matching in $B$. The directed graph $F=({\mathcal C},A)$ contains an arc $(C,C')$ for every pair of cycles
$C,C'$ in ${\mathcal C}$ such that in $M$, $C$ is matched to a node in $C'$. Now $F$ contains a spanning subgraph $F'$ consisting of 
directed paths of length $2$, stars where there are directed edges from the leaves into the center node of the star and isolated nodes.
Let $r_2$ be the number of pure cycles (i.e. consisting only of $1$-edges) which are isolated in $F'$.
Let $n_2$ denote the total number of vertices which are contained in these cycles.
Since all the cycles in ${\mathcal C}$ are of length at least four, $r_2\leq n_2\slash 4$. Moreover, $r_2\leq c_2$.
In \cite{Papadimitriou1993}, a tour $\tau$ is constructed with 
\begin{equation}
\text{cost}(\tau)\leq n+k+\frac{2}{9}\left (n-n_2-k\right )+r_2
\end{equation}
Thus we obtain 
\begin{align*}
\text{cost}(\tau) & \leq  n+k+\frac{2}{9}\left (n-n_2-k\right )+\frac{n_2}{4}\\
                  & =     \frac{11}{9}n+\frac{7}{9}k+\frac{9-8}{36}n_2 = \frac{11}{9}n+\frac{7}{9}k+\frac{1}{36}c_2
\end{align*}
Let $\tau^*$ denote an optimum tour. Then
\begin{equation}
\text{cost}(\tau^*)\geq\max\{n+k,n+c_2\}
\end{equation}
We observe that for every $x\in [0,1]$, $\max\{n+k,n+c_2\}\geq (1-x)(n+k)+x(n+c_2)$. This yields the following bound on the approximation ratio
of this algorithm (for every $x\in [0,1]$):
\begin{align*}
\frac{\text{cost}(\tau )}{\text{cost}(\tau^*)} & \leq \frac{\frac{11}{9}n+\frac{7}{9}k+\frac{1}{36}c_2}{n+x(c_2-k)+k}
\end{align*}
Directly from the definition of $c_2$ we obtain that $c_2\leq k$.
Therefore we choose $x=0$ and obtain
\[\frac{\text{cost}(\tau )}{\text{cost}(\tau^*)}\leq\frac{\frac{11}{9}n+\left (\frac{7}{9}+\frac{1}{36} \right )k}{n+k}\]
The right hand side of this inequality is monotone decreasing in $k$.
Recall that $k$ is the number of $2$-edges in a minimum length cycle cover. 

For $\beta>1$ we have $n=\zeta (\beta )e^{\alpha}$. Moreover, the number of $2$-edges
in the optimum cycle cover is at least half the number of degree $1$ nodes in $G$. This yields
\begin{equation}
\frac{\text{cost}(\tau )}{\text{cost}(\tau^*)}\leq\frac{\frac{11}{9}\zeta (\beta )+\frac{29}{36}\cdot\frac{1}{2}}{\zeta (\beta )+\frac{1}{2}}
 = \frac{\frac{11}{9}\zeta (\beta )+\frac{29}{72}}{\zeta (\beta )+\frac{1}{2}}
\end{equation}
and thus the following theorem
\begin{theorem}\label{thm:DeterministicPLG}
 For $\beta>1$, there exists an approximation algorithm for $(1,2)$-TSP and underlying $(\alpha,\beta)$-PLG with approximation ratio $\frac{\frac{11}{9}\zeta (\beta )+\frac{29}{72}}{\zeta (\beta )+\frac{1}{2}}$.
\end{theorem}

In particular, for $(\alpha,\beta)$-PLGs with power law exponent $\beta>1.666$, the above theorem yields an improvement over current best general $(1,2)$-TSP bound of $\frac{8}{7}$ due to \cite{Berman2006} (see \autoref{fig:DetRatio}).



\subsection{The Case of Large Exponents}\label{sec:AsymptoticUpper}

Next we consider the asymptotic behaviour of the algorithm in \cite{Papadimitriou1993} for large power law exponents $\beta\to\infty$.
Recall that, for $\beta> 2$, $\zeta(\beta)\e^{\alpha}$ is the total number of nodes and $\e^{\alpha}$ is the number of degree $1$ nodes. Since $\zeta(\beta)\to 1$ as $\beta\to \infty$, at some point the number of degree $1$ nodes dominates the total number of nodes.
This enforces that some fraction of degree $1$ nodes must be adjacent to other degree $1$ nodes.

Let $m_1$ be the number of degree $1$ nodes which are adjacent to another node of degree $1$.
We have that
\begin{align*}
 m_1 & \geq  \e^\alpha - \text{vol}([2,\Delta]) = \e^\alpha - \left(\zeta(\beta-1)\e^\alpha -\e^\alpha\right) = (2-\zeta(\beta-1))\e^\alpha 
\end{align*}
The right hand side of this inequality is strictly positive for $\beta>2.729$. 
Since every degree $1$ node contributes an amount of $1.5$ to the cost of a tour and the contribution of every other node is at least one, the cost of any tour is at least $\e^\alpha\cdot 1.5 + (\zeta(\beta)-1)\e^\alpha$.
On the other hand, we can construct a tour containing $\frac{m_1}{2}$ edges connecting degree $1$ nodes by first contracting all these edges and then applying any TSP algorithm to the residual instance.
Since the contribution of every other node is at most $2$, we obtain an approximation ratio of
\begin{align*}
\frac{\frac{3}{2} m_1 + 2(\zeta(\beta)\e^\alpha - m_1)}{\frac{3}{2}\e^\alpha + (\zeta(\beta)-1)\e^\alpha} & = \frac{2\zeta(\beta)\e^\alpha - \frac{1}{2} m_1}{\frac{3}{2}\e^\alpha + (\zeta(\beta)-1)\e^\alpha}\\
 & \leq \frac{2\zeta(\beta)\e^\alpha - \frac{1}{2}(2-\zeta(\beta-1))\e^\alpha}{\frac{3}{2}\e^\alpha + (\zeta(\beta)-1)\e^\alpha}\\
  & = \frac{2\zeta(\beta) + \frac{1}{2}\zeta(\beta-1) -1}{\zeta(\beta)+ \frac{1}{2}}
\end{align*}
This yields the following theorem.
\begin{theorem}\label{thm:LargeExponents}
 For $\beta>2.729$, there exists an approximation algorithm for $(1,2)$-TSP and underlying $(\alpha,\beta)$-PLG with approximation ratio $\frac{2\zeta(\beta) + \frac{1}{2}\zeta(\beta-1) -1}{\zeta(\beta)+ \frac{1}{2}}$.
\end{theorem}

In \autoref{fig:DetRatio} we plot the approximation ratios for the deterministic case and the case of large exponents in comparison with the current best general upper bound of $\frac{8}{7}$ due to \cite{Berman2006}.
We observe that an improvement over $\frac{8}{7}$ is achieved for $\beta>3.765$ and an improvement over the bound $a(\beta)$ of \autoref{thm:DeterministicPLG} is achieved for $\beta>4.309$.

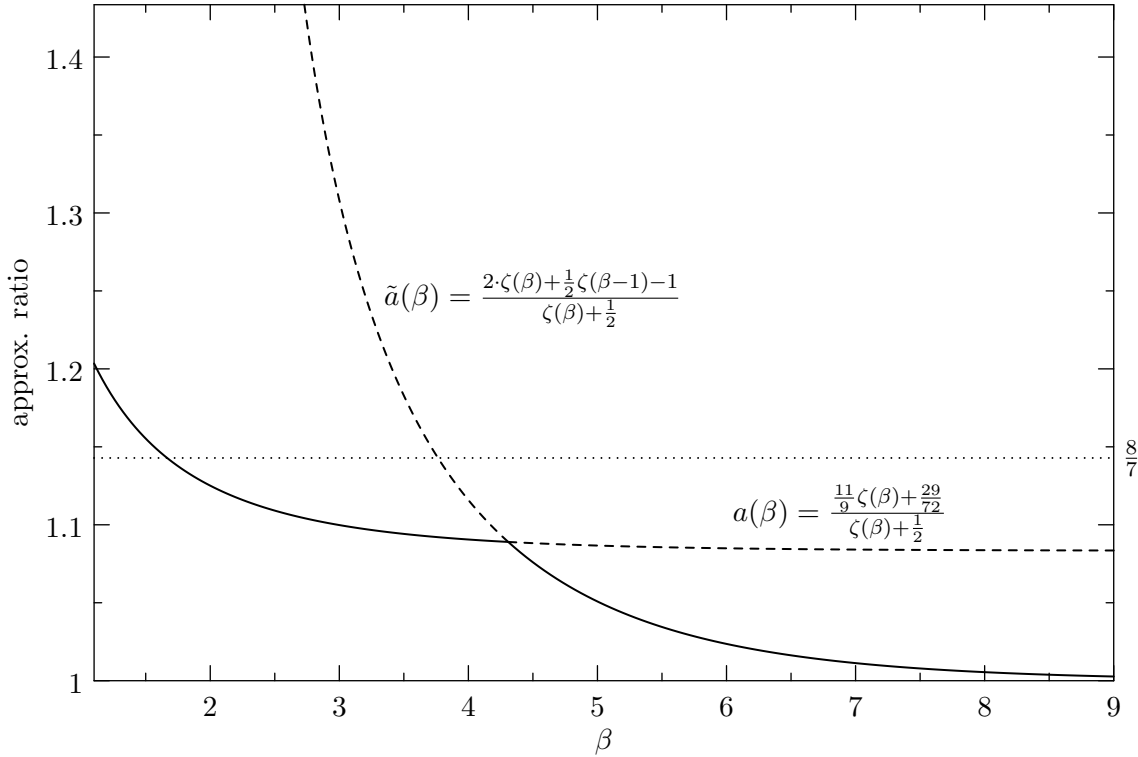
\begin{figure}[htb]
 \input{DetRatio}
\caption{Approximation ratios for the Power Law $(1,2)$-TSP as of \autoref{thm:DeterministicPLG}, $a(\beta)=\frac{\frac{11}{9}\zeta (\beta )+\frac{29}{72}}{\zeta (\beta )+\frac{1}{2}}$, and \autoref{thm:LargeExponents} ($\tilde{a}(\beta)=\frac{2\zeta(\beta) + \frac{1}{2}\zeta(\beta-1) -1}{\zeta(\beta)+ \frac{1}{2}}$) in comparison with the upper bound of $\frac{8}{7}$ due to \cite{Berman2006}.}
\label{fig:DetRatio}
\end{figure}

\subsection{Improved Upper Bounds for Random PLGs}\label{sec:RandomUpper}
Now we consider the case when the underlying graph $G$ is a random PLG. At first we present our general idea that will then yield improved approximation ratios for $(1,2)$-TSP with underlying random $(\alpha,\beta)$-PLG.

In the case of deterministic PLGs we have already made use of the following observation. The cost of a tour $\tau$ can be written as the sum of node costs $c_{\tau}(v)$, where
$c_{\tau}(v)$ is defined as $\frac{1}{2}$ times the cost of edges incident to $v$ in $\tau$. In the case of the $(1,2)$-TSP, $c_{\tau}(v)\geq 1$, since any of the two incident
edges has cost at least $1$. Moreover, if $v$ is a node of degree $1$, then $c_{\tau}(v)\geq \nicefrac{3}{2}$, since at least one of the two edges incident to $v$ has cost $2$.
This yields a lower bound  $k\geq\frac{\e^{\alpha}}{2}$ (see \autoref{sec:DeterministicUpper}). 
Now we want to improve this lower bound by taking into account the \emph{expected number} of low degree 
neighbors of nodes in an underlying random PLG. 
Let $\Gamma_i(v)$ be the set of degree $i$ neighbors of node $v$ and $N_i(v)=|\Gamma_i(v)|$ its cardinality.
Let $N_i(V)=\sum_{v\in V}N_i(v)$. 
Now we let $A_1(v)$ denote the additional lower bound of $k$ generated by degree $1$ neighbors of $v$, namely 
$A_1(v)=\min_{\tau}\sum_{u\in\Gamma_1(v)}(c_{\tau}(u)-\nicefrac{3}{2})$, where the minimum is over all tours $\tau$.
For a set $U\subseteq V$ of nodes we write  $A_1(U)=\sum_{v\in U}A_1(v)$. 
Below we will show that $A_1(v)\geq\frac{N_1(v)}{2}-1$. We will also take into account the additional lower bound on the tour cost generated by degree $2$ neighbors of a node $v$.
This lower bound is inherently tour-dependent and will be denoted as $A_{2,\tau}(v)$, where $v$ is a node of degree $>2$ and $\tau$ is a tour.
It turns out that the values $A_{2,\tau}(v)$ are not independent from each other, and we will only be able to give a lower bound on the sum $A_{2,\tau}(V)$.


Let us first consider the case $i=1$. For the degree $1$ nodes we already accounted a value of $\frac{1}{2}$ for each node for our lower bound $k\geq\frac{\e^{\alpha}}{2}$. Now suppose that a node $v$ of degree $i$ has $j\leq i$ degree $1$ neighbors. This provides an additional value of $\frac{1}{2}$ for $j-2$ neighbors since any tour $\tau$ uses at most two $1$-edges connecting $v$. 
The remaining $j-2$ neighbors are connected via $2$-edges.
In \autoref{fig:StarOne} we depict the situation around a node $v$ with $j$ degree $1$ neighbors. 

\begin{figure}[htb]
 \centering
 \includegraphics[scale=1]{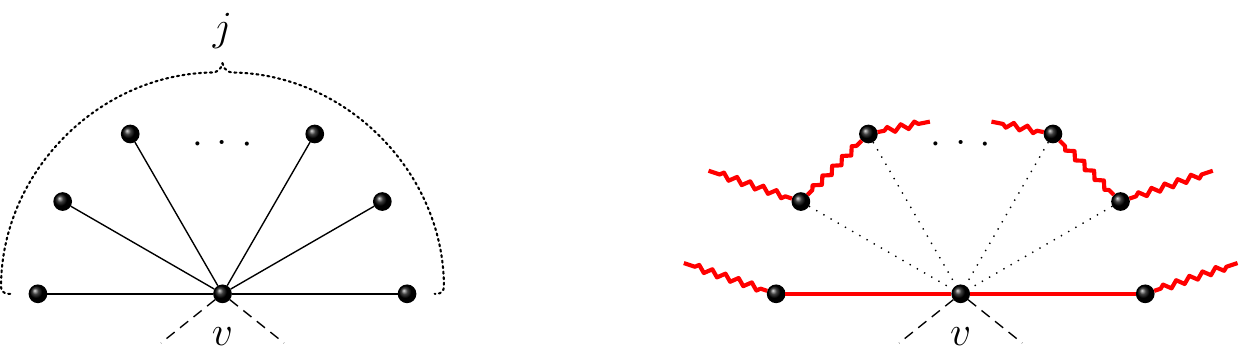}
 \caption{(left) Node $v$ with $j$ degree $1$ neighbors. (right) A possible tour around $v$ which enters and leaves $v$ via degree $1$ neighbors and matches the remaining degree $1$ vertices via $2$-edges. Since we already accounted $\frac{1}{2}$ for every degree $1$ vertex, this introduces an additional value of $A_1(v)=\frac{j}{2}-1$.}
 \label{fig:StarOne}
\end{figure}

Now we turn to the case $i=2$. We will first give a definition of the terms $A_{2,\tau}(v)$. Then we will prove that for every tour $\tau$ the two inequalities
$c(\tau)\geq n+e^{\alpha}\slash 2+A_1(V)+A_{2,\tau}(V)$ and $A_{2,\tau}(V)\geq\sum_{\text{deg}(v)>2}\frac{1}{2}\cdot (N_2(v)-2)$ hold.
Recall that for a tour $\tau$ and a node $v\in V$, the tour cost of $v$, $c_{\tau}(v)$ is defined as the sum of half of all the edge costs of tour edges incident to $v$. 
The set of neighbors of $v$ in $G$ is denoted as $\Gamma_G(v)$. Moreover, we let $\Gamma_{\tau}(v)$ denote the set of nodes which are adjacent to $v$ in the tour $\tau$.
$\Gamma_{2,\tau}(v)$ denotes the set of nodes of degree $2$ in $G$ which are adjacent to $v$ in $\tau$. For nodes $v\in V$ of degree $\text{deg}(v)>2$ we let
\[A_{2,\tau}(v)\: =\: \sum_{u\in \Gamma_G(v)\setminus\Gamma_{2,\tau}(v)}\frac{c_{\tau}(u)-1}{|\{v'|\text{deg}(v')>2,\: v'\in \Gamma_G(u)\setminus\Gamma_{\tau}(u)\}|},\]
and $A_{2,\tau}(v)=0$ otherwise.
The following lemma shows that using the values $A_1(v),A_{2,\tau}(v)$, we obtain a lower bound for the cost of a given tour $\tau$.
\begin{lemma}
For every tour $\tau$, $c(\tau)\geq n+\frac{e^{\alpha}}{2}+A_1(V)+A_{2,\tau}(v)$.
\end{lemma}
\begin{proof}
We let $A_{1,\tau}(v)=\sum_{u\in\Gamma_1(v)}(c_{\tau}(u)-\nicefrac{3}{2})$. We observe that for every node $v$ in $G$, $c_{\tau}(v)\geq 1$, and for every node $u$ of degree $1$,
$c_{\tau}(u)\geq\nicefrac{3}{2}$. We obtain
\begin{align*}
c(\tau ) & = \sum_{v}c_{\tau}(v)\\
         & \geq n +\frac{e^{\alpha}}{2}+\sum_{\text{deg}(u)=1}(c_{\tau}(u)-\nicefrac{3}{2}) + \sum_{\text{deg}(u)=2}(c_{\tau}(u)-1)\\
         & \geq n +\frac{e^{\alpha}}{2}+\sum_{\text{deg}(v)>2}\sum_{u\in\Gamma_1(v)}(c_{\tau}(u)-\nicefrac{3}{2}) + \sum_{\text{deg}(u)=2}(c_{\tau}(u)-1)
\end{align*}
We observe that for every node $u$ of degree $2$, 
\[c_{\tau}(u)-1=\sum_{v\in\Gamma_G(u)\setminus\Gamma_{\tau}(u),\text{deg}(v)>2}\frac{c_{\tau}(u)-1}{|\{v'|\text{deg}(v')>2,\: v'\in \Gamma_G(u)\setminus\Gamma_{\tau}(u)\}|}\]
This concludes the proof of the lemma.
\end{proof} 
Now assume that a node $v$ of degree $>2$ has $l$ degree $2$ neighbors. Again, any tour $\tau$ uses at most two $1$-edges while visiting $v$. 
The remaining $l-2$ neighbors may be entered via $1$-edges but then only left via $2$-edges for which we may account an additional value of $\frac{1}{2}$. 
In \autoref{fig:StarTwo} we depicted the situation around a node $v$ with $l$ neighbors of degree $2$.

\begin{figure}[htb]
 \centering
 \includegraphics[scale=1]{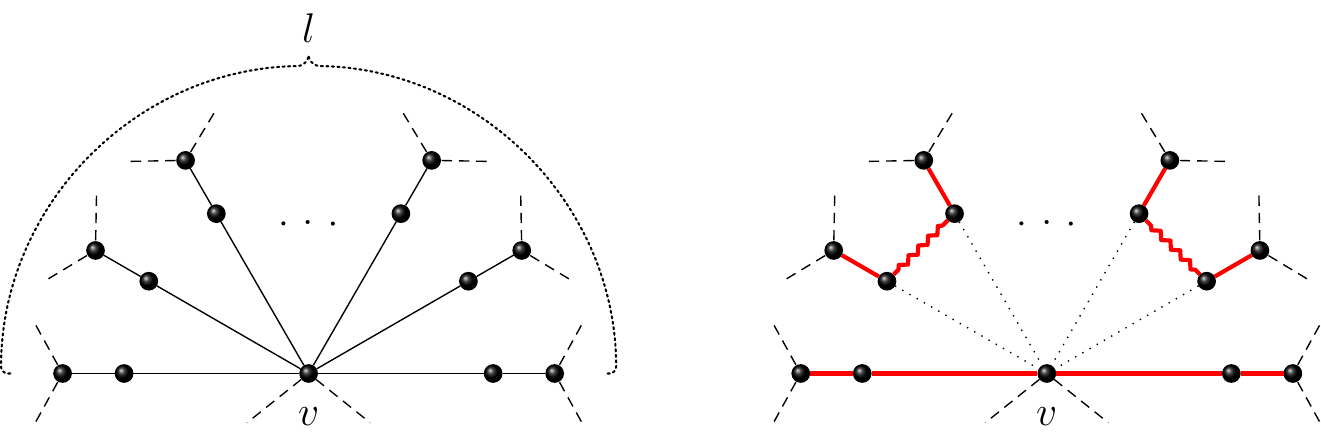}
 \caption{(left) Node $v$ with $l$ degree $2$ neighbors. (right) A possible tour around $v$ which enters and leaves $v$ via degree $2$ neighbors and matches the remaining degree $2$ vertices pairwise via $2$-edges. This introduces an additional value of $A_2(v)=\frac{l}{2}-1$.}
 \label{fig:StarTwo}
\end{figure}
We will now prove the following lower bound on the contribution of degree $2$ neighbors to the cost of a tour.
\begin{lemma}
For every tour $\tau$, $A_{2,\tau}(V)=\sum_{v,\text{deg}(v)>2}A_{2,\tau}(v)\geq\sum_{\text{deg}(v)>2}\frac{1}{2}\cdot (N_2(v)-2)$.
\end{lemma}
\begin{proof}
Consider a pair of nodes $v,u$ with $\text{deg}(v)>\text{deg}(u)=2$ such that $u$ and $v$ are adjacent in $G$ but not in $\tau$.
Let $v'$ be the other neighbor of $u$. Since $v$ is not adjacent to $u$ in the tour $\tau$, the node $u$ is incident to at least one tour edge of cost $2$.
Now we consider three different cases concerning the other neighbor $v'$.
If $v'$ is a node of degree $\leq 2$, then we do not have to charge anything to $v'$, and thus we can charge $c_{\tau}(u)-1\geq \nicefrac{3}{2}-1=\nicefrac{1}{2}$ to $v$.
If $v'$ is a node of degree $>2$ which is adjacent to $u$ within the tour $\tau$, we can also charge $c_{\tau}(u)-1\geq \nicefrac{1}{2}$ to $v$.
In the remaining case, $v'$ is a node of degree $>2$ which is not adjacent to $u$ in $\tau$. But then $u$ is incident to two edges of cost $2$ in $\tau$, which implies
$c_{\tau}(u)-1=2-1=1$, and we can charge $\nicefrac{1}{2}$ to $v$ and $\nicefrac{1}{2}$ to $v'$. The same holds when $v=v'$ and $u$ is not adjacent to $v$ in $\tau$. 
Finally we observe that for every node $v$ of degree $>2$, all but at most two degree $2$ neighbors of $v$ are not adjacent to $v$ in $\tau$. This concludes the 
proof of the lemma.
\end{proof}
In the following we quantify the effect of the above considerations and the resulting upper approximation bounds for the case of $(1,2)$-TSP and underlying random $(\alpha,\beta)$-PLGs. We present separate analyses for the three subcases $\beta>2$, $\beta=2$ and $1<\beta<2$.

\paragraph{Resulting Upper Bounds for $\beta>2$}

We are going to estimate the value $A_1(V)$ for the case that the underlying graph is an $(\alpha,\beta)$ power law graph $\gab=(V,E)$ and for power law exponents $\beta>2$. We show the following lemma.
\begin{lemma}\label{lem:Ek_One}
 Let $\beta> 2$ and $\gab=(V,E)\in\Gab$ be a random PLG in the $\Mab$ model. 
 We have that 
 \[\E[A_1(V)] \geq \frac{\e^{\alpha}}{\zeta(\beta-1)^{\beta-1}2^{\beta-1}}\left(\frac{1}{(\beta-2)(\beta-1)}-\frac{1}{4\zeta(\beta-1)}\right).\]
\end{lemma}
\begin{proof}

First we give a lower bound for the expected number of degree $1$ neighbors of a node $v$ of degree $i$.
\begin{align*}
 \E[N_1(v)] &= \sum_{j=1}^{i}\Pr(\text{$j$-th neighbor has deg. $1$})\\
  &= \sum_{j=1}^{\e^{\alpha}}\Pr(\text{the $j$-th deg. $1$ node will be adjacent to $v$})\\
  &= \e^{\alpha}\frac{i}{\zeta(\beta-1)\e^{\alpha}-1} = \frac{i}{\zeta(\beta-1)-\e^{-\alpha}} \geq \frac{i}{\zeta(\beta-1)}
\end{align*}
Thus, the expected additional value for a node $v$ of degree $i$ contributing to $k$ is $\E[A_1(v)]=\frac{i}{2\zeta(\beta-1)}-1$. This value is positive for $i\geq 2\zeta(\beta-1)$.
Summing up over all vertices with sufficiently large degrees, we obtain 
\begin{align}
 \E[A_1(V)] &\geq \sum_{i=2\zeta(\beta-1)}^{\Delta}\frac{\e^{\alpha}}{i^{\beta}}\left(\frac{i}{2\zeta(\beta-1)}-1\right)\label{eq:EA_11}
\end{align}
Replacing the sums in (\ref{eq:EA_11}) by the corresponding integral, we get the following lower bound on $\E[A_2(V)]$:
\begin{align*}
 \E[A_1(V)] &\geq \frac{\e^{\alpha}}{2\zeta(\beta-1)}\left(\int_{2\zeta(\beta-1)}^{\Delta}\frac{1}{x^{\beta-1}} dx + \frac{1}{2}\left(\frac{1}{(2\zeta(\beta-1))^{\beta-1}}-\frac{1}{\Delta^{\beta-1}}\right)\right)\\
  &\quad - \e^{\alpha}\left(\int_{2\zeta(\beta-1)}^{\Delta}\frac{1}{x^{\beta}}+\frac{1}{(2\zeta(\beta-1))^{\beta}}-\frac{1}{\Delta^{\beta}}\right) + \frac{\e^{\alpha}}{\Delta^{\beta}\left(\frac{\Delta}{2\zeta(\beta-1)}-1\right)}
   \end{align*}
The anti-derivative now has the following form
   \begin{align*}
  &\E[A_1(V)] \geq \frac{\e^{\alpha}}{2\zeta(\beta-1)}\left(\left[\frac{x^{2-\beta}}{2-\beta}\right]_{2\zeta(\beta-1)}^{\Delta}+\frac{1}{2}\left(\frac{1}{(2\zeta(\beta-1))^{\beta-1}}-\frac{1}{\Delta^{\beta-1}}\right)\right)\\
  &\qquad -  \e^{\alpha}\left(\left[\frac{x^{1-\beta}}{1-\beta}\right]_{2\zeta(\beta-1)}^{\Delta}+\frac{1}{(2\zeta(\beta-1))^{\beta}}-\frac{1}{\Delta^{\beta}}\right) + \frac{\Delta}{2\zeta(\beta-1)}-1\\
  &= \frac{\e^{\alpha}}{2\zeta(\beta-1)}\left(\frac{1}{(\beta-2)(2\zeta(\beta-1))^{\beta-2}}-\frac{1}{(\beta-2)\Delta^{\beta-2}}+\frac{1}{2^{\beta}\zeta(\beta-1)^{\beta-1}}-\frac{1}{2\Delta^{\beta-1}}\right)\\
  &\qquad - \e^{\alpha}\left(\frac{1}{(\beta-1)(2\zeta(\beta-1))^{\beta-1}}-\frac{1}{(\beta-1)\Delta^{\beta-1}}+\frac{1}{2^{\beta}\zeta(\beta-1)^{\beta}}-\frac{1}{\Delta^{\beta}}\right) + o(\e^{\alpha})\\
  &= \frac{\e^{\alpha}}{\zeta(\beta-1)^{\beta-1}2^{\beta-1}}\left(\frac{1}{\beta-2}+\frac{1}{4\zeta(\beta-1)}-\frac{1}{\beta-1}-\frac{1}{2\zeta(\beta-1)} - o(1)\right)\\
  &= \frac{\e^{\alpha}}{\zeta(\beta-1)^{\beta-1}2^{\beta-1}}\left(\frac{1}{(\beta-2)(\beta-1)}-\frac{1}{4\zeta(\beta-1)}\right)
\end{align*}
This concludes the proof of the lemma.
\end{proof}

Next, we want to further improve the expected lower bound on $k$ similar to \autoref{lem:Ek_One} by taking into account the additional value generated by degree $2$ neighbors of nodes in the underlying graph $\gab$.

Let $A_2(V)$ be the value added by degree $2$ neighbors of nodes in the underlying graph $G=(V,E)$. We prove the following lemma.
\begin{lemma}\label{lem:Ek_Two}
 Let $\gab=(V,E)\in\Gab$ be a random PLG in the $\Mab$ model. 
 We have that 
 \[\E[A_2(V)]\geq (1-o(1))\e^{\alpha}\left(\frac{2^{-\beta(\beta-1)}\zeta(\beta-1)^{1-\beta}}{(\beta-1)(\beta-2)}+\frac{1}{\zeta(\beta-1)^{\beta}}\left(\frac{1}{2^{\beta\cdot\beta+1}}-\frac{1}{2^{2\beta}}\right)\right).\]
\end{lemma}
\begin{proof}
 First we give a lower bound for the expected number of degree $2$ neighbors of a node $v$ of degree $i$. 
 \begin{align*}
  \E[N_2(v)] &= \sum_{j=1}^{i}\Pr(\text{$j$-th neighbor has deg. $2$})\\
  &= \sum_{j=1}^{\frac{\e^{\alpha}}{2^{\beta}}}\Pr(\text{the $j$-th deg. $2$ node will be adjacent to $v$})\\
  &= \frac{\e^{\alpha}}{2^{\beta}}\left(\frac{i}{\zeta(\beta-1)\e^{\alpha}-1} + \frac{\zeta(\beta-1)\e^{\alpha}-1-i}{\zeta(\beta-1)\e^{\alpha}-1} \cdot \frac{i}{\zeta(\beta-1)\e^{\alpha}-2}\right)\\
  &> \frac{\e^{\alpha}}{2^{\beta}} \cdot \frac{2i}{\zeta(\beta-1)\e^{\alpha}} = \frac{i}{2^{\beta-1}\zeta(\beta-1)}
 \end{align*}
 Now, the expected additional value for a node $v$ of degree $i$ contributing to $k$ is $\E[A_2(v)]=\frac{i}{2^{\beta}\zeta(\beta-1)}-1$. This value is positive for $i\geq 2^{\beta}\zeta(\beta-1)$.
Summing up over all vertices with sufficiently large degrees, we obtain 
\begin{align}
  \E[A_2(V)] &\geq \sum_{i=2^{\beta}\zeta(\beta-1)}^{\Delta} \frac{\e^{\alpha}}{i^{\beta}}\left(\frac{i}{2^{\beta}\zeta(\beta-1)}-1\right)\nonumber\\
   &= \frac{\e^{\alpha}}{2^{\beta}\zeta(\beta-1)} \sum_{i=2^{\beta}\zeta(\beta-1)}^{\Delta} \frac{1}{i^{\beta-1}}-\e^{\alpha} \sum_{i=2^{\beta}\zeta(\beta-1)}^{\Delta} \frac{1}{i^{\beta}}.\label{eq:EA_21}
\end{align} 
Replacing the sums in (\ref{eq:EA_21}) by the corresponding integral, we obtain the inequality
 \begin{align*}
  \E[A_2(V)] &\geq \frac{\e^{\alpha}}{2^{\beta}\zeta(\beta-1)} \left(\frac{1}{\Delta^{\beta-1}}+\int_{2^{\beta}\zeta(\beta-1)}^{\Delta}x^{1-\beta}dx+\frac{1}{2}\left(\frac{1}{(2^{\beta}\zeta(\beta-1))^{\beta-1}}-\frac{1}{\Delta^{\beta-1}}\right)\right) \\
   &\qquad - \e^{\alpha} \left(\frac{1}{\Delta^{\beta}}+\int_{2^{\beta}\zeta(\beta-1)}^{\Delta}x^{-\beta}dx+\frac{1}{(2^{\beta}\zeta(\beta-1))^{\beta}}-\frac{1}{\Delta^{\beta}}\right)
   \end{align*}
The anti-derivative now has the following form
   \begin{align*}
 \E[A_2(&V)] \geq \frac{\e^{\alpha}}{2^{\beta}\zeta(\beta-1)} \left(\left[\frac{x^{2-\beta}}{2-\beta}\right]_{2^{\beta}\zeta(\beta-1)}^{\Delta}+\frac{1}{2^{\beta(\beta-1)+1}\zeta(\beta-1)^{\beta-1}}+\frac{1}{2\Delta^{\beta-1}}\right)\\
   &\qquad -\e^{\alpha} \left(\left[\frac{x^{1-\beta}}{1-\beta}\right]_{2^{\beta}\zeta(\beta-1)}^{\Delta}+\frac{1}{2^{2\beta}\zeta(\beta-1)^{\beta}}\right)\\
   &= \frac{\e^{\alpha}}{2^{\beta}\zeta(\beta-1)} \left(\frac{(2^{\beta}\zeta(\beta-1))^{2-\beta}-\Delta^{2-\beta}}{\beta-2}+\frac{1}{2^{\beta(\beta-1)+1}\zeta(\beta-1)^{\beta-1}}+\frac{1}{2\Delta^{\beta-1}}\right)\\
   &\qquad -\e^{\alpha} \left(\frac{(2^{\beta}\zeta(\beta-1))^{1-\beta}-\Delta^{1-\beta}}{\beta-1}+\frac{1}{2^{2\beta}\zeta(\beta-1)^{\beta}}\right)\\
   &= \e^{\alpha} \left(\frac{1}{(\beta-2)2^{\beta(\beta-1)}\zeta(\beta-1)}+\frac{1}{2^{\beta\cdot\beta+1}\zeta(\beta-1)^{\beta}}-\frac{1}{(\beta-1)2^{\beta(\beta-1)}\zeta(\beta-1)}\right. \\
   &\qquad \left.-\frac{1}{2^{2\beta}\zeta(\beta-1)^{\beta}} - o(1)\right)\\
   &= (1- o(1))\e^{\alpha} \left(\frac{2^{-\beta(\beta-1)}\zeta(\beta-1)^{1-\beta}}{(\beta-1)(\beta-2)}+\frac{1}{\zeta(\beta-1)^{\beta}}\left(\frac{1}{2^{\beta\cdot\beta+1}}-\frac{1}{2^{2\beta}}\right)\right)
 \end{align*}
 This concludes the proof of the lemma.
\end{proof}

From \autoref{lem:Ek_One} and \autoref{lem:Ek_Two} we obtain the following corollary.
\begin{corollary}
Let $k$ be the number of $2$-edges in the optimum cycle cover of a $(1,2)$-TSP instance with underlying power law graph $\gab\in\Gab$ drawn from $P(\alpha,\beta )$. We have
 \begin{align*}
  \E[k] &\geq \frac{\e^{\alpha}}{2} + \E[A_1(V)] + \E[A_2(V)]\\
  &\geq \frac{\e^{\alpha}}{2} + \frac{\e^{\alpha}}{\zeta(\beta-1)^{\beta-1}2^{\beta-1}}\left(\frac{1}{(\beta-2)(\beta-1)}-\frac{1}{4\zeta(\beta-1)}\right)\\
  &\qquad + \e^{\alpha}\left(\frac{2^{-\beta(\beta-1)}\zeta(\beta-1)^{1-\beta}}{(\beta-1)(\beta-2)}+\frac{1}{\zeta(\beta-1)^{\beta}}\left(\frac{1}{2^{\beta\cdot\beta+1}}-\frac{1}{2^{2\beta}}\right)\right) =: E_k
 \end{align*}
\end{corollary}

From the above corollary we finally obtain
\begin{theorem}\label{thm:Beta>2}
  For $\beta>2$, there exists an approximation algorithm for $(1,2)$-TSP and underlying \emph{random} $(\alpha,\beta)$-PLG with expected 
approximation ratio $\frac{\frac{11}{9}\zeta(\beta)+\frac{29}{36}\cdot E_k}{\zeta(\beta)+E_k}$.
\end{theorem}

In \autoref{fig:RandomRatio1} we compare the expected approximation ratio $\hat{a}(\beta)$ of \autoref{thm:Beta>2} to the ratio $a(\beta)$ achieved in the deterministic setting (\autoref{thm:DeterministicPLG}).

\begin{figure}[htb]
  \input{RandomRatio}
  \caption{Plot of the expected approximation ratio $\hat{a}(\beta)$ of \autoref{thm:Beta>2} as compared to the ratio $a(\beta)$ achieved in the deterministic setting (\autoref{thm:DeterministicPLG}).}
  \label{fig:RandomRatio}
\end{figure}
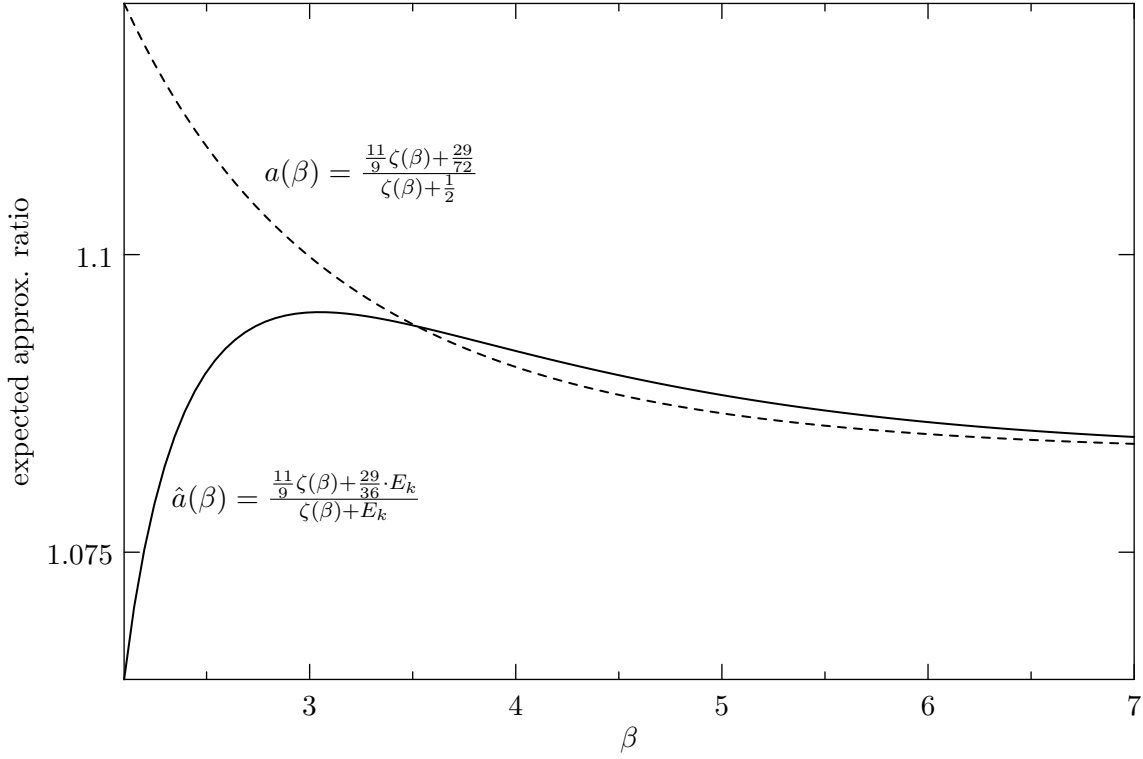

\paragraph{Resulting Upper Bound for $\mathbf{1<\beta<2}$}

We consider now the case when the power law exponent $\beta$ satisfies $1<\beta<2$.
Again we will give estimates for the expected additional contribution of degree $1$ and degree $2$ nodes $A_1(V)$ and $A_2(V)$. 
The contribution of degree $1$ nodes is given in the following lemma.
\begin{lemma}\label{lem:Ek_One2}
 Let $1<\beta<2$ and $\gab=(V,E)\in\Gab$ be a random PLG in the $\Mab$ model. 
 We have that 
 $\E[A_1(V)] \geq \e^{\alpha}$.
\end{lemma}
The proof is similar to that of \autoref{lem:Ek_One}.
The expected number of degree $1$ neighbors of a node $v$ can be lower bounded by $\e^\alpha\frac{\text{deg}(v)}{2|E|-1}\geq(2-\beta)\text{deg}(v)\e^{\alpha(1-\nicefrac{2}{\beta})}$.
Based on this we obtain
\[\E[A_1(V)]\geq \frac{\e^\alpha}{2}+\sum_{\frac{4}{2-\beta}\e^{\alpha(\nicefrac{2}{\beta}-1)}}^\Delta \frac{2-\beta}{2} j^{1-\beta} \e^{\alpha(2-\nicefrac{2}{\beta})}.\]
We estimate the sums by integrals and finally obtain $\E[A_1(V)]\geq\e^\alpha$.

The case of degree $2$ nodes being adjacent to nodes of higher degree is considered in the following lemma.
\begin{lemma}\label{lem:Ek_Two2}
 Let $1<\beta<2$ and $\gab=(V,E)\in\Gab$ be a random PLG in the $\Mab$ model. 
 We have that 
 $\E[A_2(V)] \geq \frac{1}{4}\e^{\alpha}$.
\end{lemma}
\begin{proof}
 We start by estimating the expected number of degree $2$ neighbors of high degree nodes in $\gab$. Suppose $v$ is a node of degree $i$.
 Then we have
 \begin{align*}
  \E[N_2(v)] 
  &= \frac{\e^{\alpha}}{2^{\beta}}\left(\frac{i}{\frac{\e^{2\alpha/\beta}}{2-\beta}-1} + \frac{\frac{\e^{2\alpha/\beta}}{2-\beta}-1-i}{\frac{\e^{2\alpha/\beta}}{2-\beta}-1} \cdot \frac{i}{\frac{\e^{2\alpha/\beta}}{2-\beta}-2}\right)\\
  &> \frac{(2-\beta)\e^{\alpha}}{2^{\beta}} \cdot \frac{i}{\e^{2\alpha/\beta}} = \frac{2-\beta}{2}\cdot\frac{i}{\e^{\alpha(\frac{2}{\beta}-1)}}
 \end{align*}
 Therefore, the expected additional contribution to $k$ generated by a node $v$ of degree $i$ in $\gab$ is
$\E[A_2(v)]=\frac{2-\beta}{4}\cdot\frac{i}{\e^{\alpha(2/\beta-1)}}-1$. This value is positive for $i\geq \e^{\alpha(2/\beta-1)}\cdot \frac{4}{2-\beta}=:\delta$. 
Summing up over all nodes with sufficiently large degrees we obtain 
 \begin{align*}
  \E[A_2(V)] &= \sum^\Delta_\delta \frac{\e^\alpha}{i^\beta}\left(\frac{2-\beta}{4}\cdot\frac{i}{\e^{\alpha(2/\beta-1)}}-1\right)\\
  &\geq \sum^\Delta_\delta \left(\frac{2-\beta}{4}\cdot\frac{i^{1-\beta}}{\e^{2\alpha(1/\beta-1)}}-\frac{\e^\alpha}{i^\beta}\right)\intertext{We replace the sum 
by the corresponding integral, we obtain}
  \E[A_2(V)] &\geq \int^{\Delta+1}_\delta x^{1-\beta}dx\cdot \frac{2-\beta}{4}\e^{2\alpha(1-1/\beta)} + \frac{2-\beta}{8}\e^{2\alpha(1-1/\beta)}\left(\left(\frac{4}{2-\beta}\e^{\alpha(2/\beta-1)}\right)^{1-\beta}\right.\\
  &\quad -(\Delta+1)^{1-\beta}\Bigg)
   -\e^{\alpha}\left(\int^{\Delta+1}_\delta x^{-\beta}dx + \left(\frac{4}{2-\beta}\e^{\alpha(2/\beta-1)}\right)^{-\beta}-(\Delta+1)^{-\beta}\right)
 \end{align*}
Again we make use of anti-derivatives in order to estimate the integral terms.
Then by rearranging terms we obtain $\E[A_2(V)] \geq \nicefrac{1}{4}\e^\alpha$.
This concludes the proof of the lemma.
 \end{proof}
 
From the above \autoref{lem:Ek_One2} and \autoref{lem:Ek_Two2} we obtain the following result.
\begin{theorem}\label{thm:1Beta2}
  For $1<\beta<2$, there exists an approximation algorithm for $(1,2)$-TSP and underlying \emph{random} $(\alpha,\beta)$-PLG with 
approximation ratio $\frac{\frac{11}{9}\zeta(\beta)+\frac{29}{36}\cdot\frac{5}{4}}{\zeta(\beta)+\frac{5}{4}}$.
\end{theorem}

\paragraph{Resulting Upper Bound for $\mathbf{\beta=2}$}

In the case when the power law exponent is $\beta=2$, we proceed similarly and obtain the following result.

\begin{theorem}\label{thm:Beta=2}
  For $\beta=2$, there exists an approximation algorithm for $(1,2)$-TSP and underlying \emph{random} $(\alpha,\beta)$-PLG with 
approximation ratio $\frac{\frac{11}{9}\zeta(\beta)+\frac{29}{36}\cdot\frac{5}{8}}{\zeta(\beta)+\frac{5}{8}}$.
\end{theorem}

\autoref{fig:RandomRatio1} shows the approximation ratio for $(1,2)$-TSP on random $(\alpha,\beta)$-PLGs for the three subcases $1<\beta<2$, $\beta=2$ and $\beta>2$.
  
\begin{figure}[htb]
  \input{RandomRatio1}
  \caption{Expected Approximation ratio for Random Power Law $(1,2)$-TSP for $1<\beta<2$ ($\tilde{a}(\beta)$, \autoref{thm:1Beta2}), $\beta=2$ (${a}(\beta)$, \autoref{thm:Beta=2}) and $\beta>2$ ($\hat{a}(\beta)$, \autoref{thm:Beta>2}).}
  \label{fig:RandomRatio1}
\end{figure}
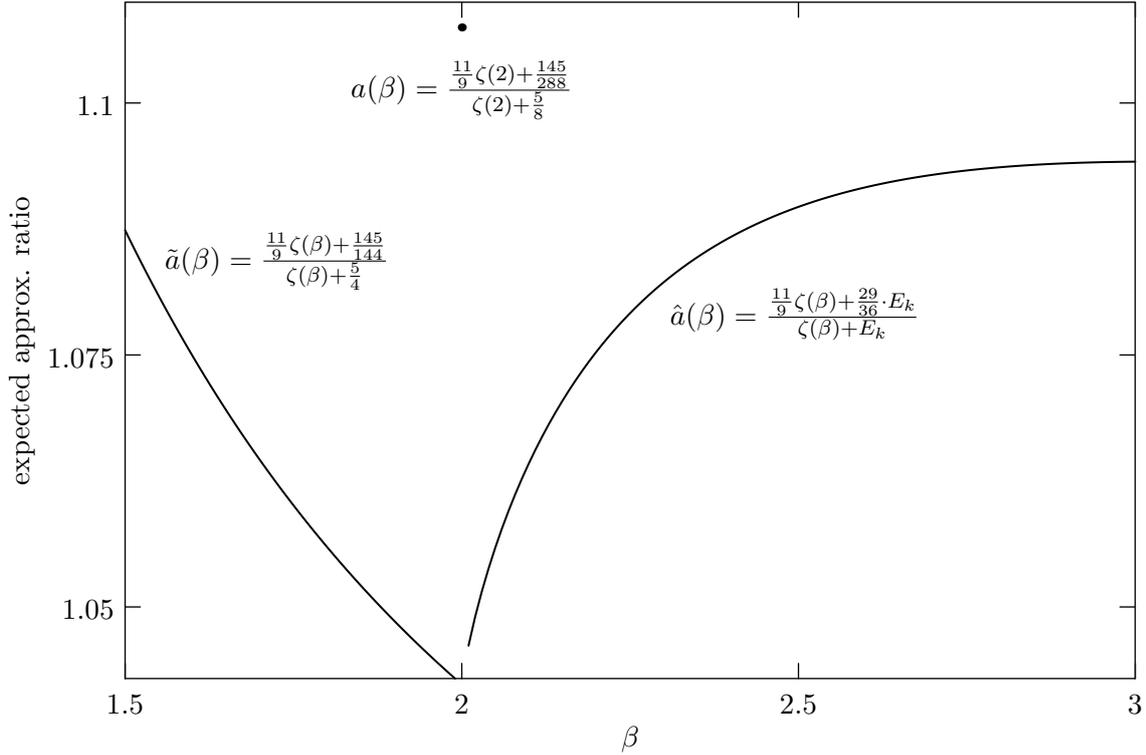

%% file: PlotGraphTSP.tex
\begin{asy}
 defaultpen(fontsize(11));
 import graph;
 import gsl;
     size(13cm,9cm,IgnoreAspect);

     pen dashed=linetype(new real[] {4,4});
     pen p=linewidth(.75);

     real f(real x) {return 2*zeta(x)/(max(zeta(x)+0.5,2));}
     real g(real x) {return (0.5+(zeta(x))/(max(2,zeta(x)+0.5)));}
     real h(real x) {return 7/5;}
 //     real m(real x) {return (2*zeta(x-1)/3+2*zeta(x)/3+5/6)/(max(2.0,zeta(x)+0.5));}
 //    real m(real x) {return (2*zeta(x-1)/3+2*zeta(x)/3+1/6)/(zeta(x)+0.5);}
     real m(real x) {return (2*zeta(x-1)/3+2*zeta(x)/3+5/6-(2/3)*((x-2)*(zeta(x-1)/2-zeta(x)+1))^((x-1)/(x-2))/(x-1))/(max(2.0,zeta(x)+0.5));}

     real mu(real x) {return (10/9+zeta(x)/(3*max(2,zeta(x)+0.5)));}
  
 //     draw("$f=\frac{2\zeta(\beta)}{\max\{\zeta(\beta)+0.5,2\}}$",graph(f,1.5,3.5),p);
     draw("MST",graph(f,1.25,2.48),p);
     draw("Chr",graph(g,1.25,2.48),p);
     draw("",graph(h,1.25,2.48),dotted+p);
 //    draw("$\frac{\frac{2}{3}\zeta(\beta -1)+\frac{2}{3}\zeta(\beta)+\frac{5}{6}-H_{\beta}}{\max\{\zeta(\beta)+\frac{1}{2},2\}}$",graph(m,2.5,3.5),dotted+p);
     draw("M-S",graph(m,2.4,2.48),dashed+p);
     draw("Mu",graph(mu,1.25,2.48),dashed+p);
     
 //    ylimits(1.968,2);
     xlimits(1.25,2.5);
     label("$\frac{7}{5}$",(2.5,7/5),E);
     xaxis("$\beta$",BottomTop,LeftTicks);
     yaxis("approx. ratio",LeftRight,RightTicks);

     //attach(legend(),(point(S).x,truepoint(S).y),10S,UnFill);
  \end{asy}

%% file: DetRatio.tex
\begin{asy}
 defaultpen(fontsize(11));
 import graph;
 import gsl;
     size(15cm,10cm,IgnoreAspect);

     pen dashed=linetype(new real[] {4,4});
     pen p=linewidth(.75);

     real g(real x) {return 8/7;}
     //real f(real x) {return 13/12;}
     real e(real x) {return (2*zeta(x) + 0.5*zeta(x-1) -1)/(zeta(x) + 0.5);}
     pair E(real x) {return (x,e(x));}
     real h(real x) {return ((11/9*zeta(x)+29/72)/(zeta(x)+1/2));}
     pair H(real x) {return (x,h(x));}
  
     draw("",graph(g,1.1,9),dotted+p,"first analysis ($\rho$)");
     //draw("",graph(f,1.1,9),dotted+p,"second analysis ($\rho'$)");
     draw("",graph(e,2.729,4.309),dashed+p,"second analysis ($\rho'$)");
     draw("",graph(h,1.1,4.309),p,"second analysis ($\rho'$)");
     draw("",graph(e,4.309,9),p,"second analysis ($\rho'$)");
     draw("",graph(h,4.309,9),dashed+p,"second analysis ($\rho'$)");
     
     ylimits(1,e(2.729));
     //ytick("$\frac{8}{7}$",(1,8/7));
     //ytick("$\frac{13}{12}$",(1,13/12));
     //ytick("$1$",(1,1));
     //labely(1,E);
     label("$\frac{8}{7}$",(9,8/7),E);
     //label("$\frac{13}{12}$",(9,13/12),E);
     label("$a(\beta)=\frac{\frac{11}{9}\zeta(\beta)+\frac{29}{72}}{\zeta(\beta)+\frac{1}{2}}$",H(6),NE);
     label("$\tilde{a}(\beta)=\frac{2\cdot\zeta(\beta) + \frac{1}{2}\zeta(\beta-1) -1}{\zeta(\beta)+ \frac{1}{2}}$",E(3.3),NE);
     xaxis("$\beta$",BottomTop,LeftTicks);
     yaxis("approx. ratio",LeftRight,RightTicks);

     //attach(legend(),(point(S).x,truepoint(S).y),10S,UnFill);
  \end{asy}

%% file: RandomRatio.tex
\begin{asy}
 defaultpen(fontsize(11));
 import graph;
 import gsl;
     size(15cm,10cm,IgnoreAspect);

     pen dashed=linetype(new real[] {4,4});
     pen p=linewidth(.75);

     //real e(real x) {return 8/7;}
     //real g(real x) {return ((11/9*zeta(x)+29/36*(1/2+(1/((x-2)*(x-1))-1/(4*zeta(x-1)))/(zeta(x-1)^(x-1)*2^(x-1)))) / (zeta(x)+1/2+(1/((x-2)*(x-1))-1/(4*zeta(x-1)))/(zeta(x-1)^(x-1)*2^(x-1))));}
     real f(real x) {return ((11/9*zeta(x)+29/36*(1/2+(1/((x-2)*(x-1))-1/(4*zeta(x-1)))/(zeta(x-1)^(x-1)*2^(x-1))+(1/((x-1)*(x-2))/(2^(x*(x-1))*zeta(x-1)^(x-1))+1/zeta(x-1)^x*(1/2^(2x+1)-1/2^(2*x))))) / (zeta(x)+1/2+(1/((x-2)*(x-1))-1/(4*zeta(x-1)))/(zeta(x-1)^(x-1)*2^(x-1))+(1/((x-1)*(x-2))/(2^(x*(x-1))*zeta(x-1)^(x-1))+1/zeta(x-1)^x*(1/2^(x^2+1)-1/2^(2*x)))));}
     pair F(real x) {return (x,f(x));}
     real h(real x) {return ((11/9*zeta(x)+29/72)/(zeta(x)+1/2));}
     pair H(real x) {return (x,h(x));}
  
     //draw("$8/7$",graph(e,2.1,7),dashed+p,"first analysis ($\rho$)");
     //draw("$g$",graph(g,2.1,7),dashed+p,"first analysis ($\rho$)");
     draw("",graph(f,2.1,7),p,"second analysis ($\rho'$)");
     draw("",graph(h,2.1,7),dashed+p,"second analysis ($\rho'$)");
     
 //    ylimits(1.968,2);
     xaxis("$\beta$",BottomTop,LeftTicks);
     yaxis("expected approx. ratio",LeftRight,RightTicks);
     
     label("$a(\beta)=\frac{\frac{11}{9}\zeta(\beta)+\frac{29}{72}}{\zeta(\beta)+\frac{1}{2}}$",H(2.75),NE);
     label("$\hat{a}(\beta)=\frac{\frac{11}{9}\zeta(\beta)+\frac{29}{36}\cdot E_k}{\zeta(\beta)+E_k}$",F(2.3),SE);

     //attach(legend(),(point(S).x,truepoint(S).y),10S,UnFill);
  \end{asy}

%% file: RandomRatio1.tex
\begin{asy}
 defaultpen(fontsize(11));
 import graph;
 import gsl;
    // size(25cm,17cm,IgnoreAspect);
     size(15cm,10cm,IgnoreAspect);

     pen dashed=linetype(new real[] {4,4});
     pen p=linewidth(.75);
     pen p2=linewidth(3);

     //real g(real x) {return ((11/9*zeta(x)+29/36*(1/2+(1/((x-2)*(x-1))-1/(4*zeta(x-1)))/(zeta(x-1)^(x-1)*2^(x-1)))) / (zeta(x)+1/2+(1/((x-2)*(x-1))-1/(4*zeta(x-1)))/(zeta(x-1)^(x-1)*2^(x-1))));}
     real f(real x) {return ((11/9*zeta(x)+29/36*(1/2+(1/((x-2)*(x-1))-1/(4*zeta(x-1)))/(zeta(x-1)^(x-1)*2^(x-1))+(1/((x-1)*(x-2))/(2^(x*(x-1))*zeta(x-1)^(x-1))+1/zeta(x-1)^x*(1/2^(2x+1)-1/2^(2*x))))) / (zeta(x)+1/2+(1/((x-2)*(x-1))-1/(4*zeta(x-1)))/(zeta(x-1)^(x-1)*2^(x-1))+(1/((x-1)*(x-2))/(2^(x*(x-1))*zeta(x-1)^(x-1))+1/zeta(x-1)^x*(1/2^(2x+1)-1/2^(2*x)))));}
     pair F(real x) {return (x,f(x));}
     //real h(real x) {return ((11/9*zeta(x)+29/72)/(zeta(x)+1/2));}
     //pair H(real x) {return (x,h(x));}
     real j(real x) {return ((11/9*zeta(x)+145/144)/(zeta(x)+5/4));}
     pair J(real x) {return (x,j(x));}
     real k(real x) {return ((11/9*zeta(x)+145/288)/(zeta(x)+5/8));}
     //real k2(real x) {return ((11/9*zeta(x)+145/288)/(zeta(x)+5/8))-0.0001;}
     pair K(real x) {return (x,k(x));}
     //real l(real x) {return (8/7);}
  
     //draw("$g$",graph(g,2.01,5),dashed+p,"first analysis ($\rho$)");
     draw("",graph(f,2.01,3),p,"second analysis ($\rho'$)");
     //draw("$h$",graph(h,2.01,5),dotted+p,"second analysis ($\rho'$)");
//     draw("$\frac{8}{7}$",graph(l,1.5,3),dotted+p,"second analysis ($\rho'$)");
     draw("",graph(j,1.5,1.99),p,"second analysis ($\rho'$)");
     draw("",graph(k,2,2),p2,"second analysis ($\rho'$)");
     //draw("",graph(k2,1.999,2.001),p2,"second analysis ($\rho'$)");
     
    ylimits(j(1.99),1.11);
  //   xlimits(0,5);
     xaxis("$\beta$",BottomTop,LeftTicks);
     yaxis("expected approx. ratio",LeftRight,RightTicks);
     
     label("${a}(\beta)=\frac{\frac{11}{9}\zeta(2)+\frac{145}{288}}{\zeta(2)+\frac{5}{8}}$",(2,1.105),S);
     label("$\tilde{a}(\beta)=\frac{\frac{11}{9}\zeta(\beta)+\frac{145}{144}}{\zeta(\beta)+\frac{5}{4}}$",J(1.55),NE);
     label("$\hat{a}(\beta)=\frac{\frac{11}{9}\zeta(\beta)+\frac{29}{36}\cdot E_k}{\zeta(\beta)+E_k}$",F(2.3),SE);

     //attach(legend(),(point(S).x,truepoint(S).y),10S,UnFill);
  \end{asy}

%% file: Lower.tex
\subsection{Lower Bounds}\label{sec:LowerBounds}
In this section we give the first approximation lower bounds for TSP in $(1,2)$-PLG instances.
The approximation hardness of the general $(1,2)$-TSP problem was proven in \cite{Papadimitriou1993}.
Later the first explicit inapproximability bounds of $\nicefrac{5381}{5380}$ was proven by \citeauthor{Engebretsen2003} \cite{Engebretsen2003}, and further improved step by step to the current best bound of $\nicefrac{123}{122}$ by \citeauthor{Karpinski2013} \cite{Karpinski2013} applying the results of \citeauthor{Berman1999} \cite{Berman1999}, \cite{Berman2001}.

Our overall approach is to construct appropriate reductions from the $(1,2)$-Metric TSP to 
$(1,2)$-metric instances whose underlying graph is a PLG.
We make use of the lower bound constructions developed in \cite{Karpinski2013b} (see also \cite{Engebretsen2006,Lampis2012,Karpinski2013a,Karpinski2013}).
The reduction in \cite{Karpinski2013b} starts from a special bounded occurrence optimization problem.

\begin{definition}[Hybrid Problem] 
 Given a system of linear equations mod $2$ over $n$ variables, $m_2$
equations with exactly two variables, and $m_3$ equations with exactly three variables, find an assignment to the variables that maximizes the number of satisfied equations.
\end{definition}

The approximation hardness of the above problem was proven by \citeauthor{Berman1999} in \cite{Berman1999} using a result of \citeauthor{Hastad2001} in \cite{Hastad2001}. They proved the following theorem.

\begin{theorem}[\cite{Berman1999}]
 For every $0<\varepsilon<\frac{1}{2}$, there exists an instance $I_H$ of the Hybrid problem with $42v$ variables, $m_2=60v$, $m_3=2v$ such that every variable occurs exactly $3$ times and it is NP-hard to decide whether $I_H$ has its optimum value above $(62-\varepsilon)v$ or below $(61+\varepsilon)v$.
\end{theorem}

The above result is used in the following reduction from Hybrid to Subcubic $(1,2)$-TSP due to \citeauthor{Karpinski2013b} \cite{Karpinski2013b}. Later we present our new reduction from \emph{slightly modified} Subcubic $(1,2)$-TSP instances to $(1,2)$-TSP instances with underlying PLG.

\paragraph{Reduction from Hybrid to Subcubic $(1,2)$-TSP}
The reduction starts from Hybrid instances with $60v$ $2$-equations, $2v$ $3$-equations and $42v$ variables. 
In \cite{Karpinski2013b}, these instances have been mapped to $(1,2)$-TSP instances $G_{SC}^{12}$ with 
\[60v\cdot 8+2v\cdot (9\cdot 8+6\cdot 3 +2\cdot 3)=v\cdot 672\]
nodes. Our aim is to modify these instances slightly in order to fit them into a power law degree distribution.
The instances $G_{SC}^{12}$
consist of building blocks called parity gadgets. A parity gadget is a graph $P$ with $8$ nodes being connected as shown in \autoref{fig:Parity}.

\begin{figure}[htb]
 \centering
 \includegraphics[scale=1]{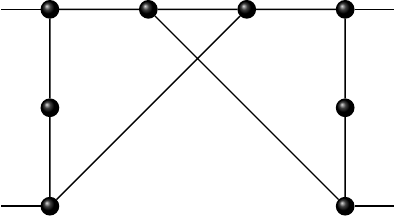}
 \caption{Parity gadget.}
 \label{fig:Parity}
\end{figure}

The graph $G_{SC}^{12}$ contains 
\begin{itemize}
\item one parity gadget for every variable,
\item one additional parity gadget for every $2$-equation,
\item for every $3$-equation a subgraph involving the parity gadgets for the variables being involved,
      additionally $6$ parity gadgets and $12$ extra nodes.
\end{itemize}

Let us describe the simulating gadget for a $3$-equation $x\oplus y\oplus z=0$, as constructed in \cite{Karpinski2013b}. Given an equation of the form $x\oplus y\oplus z=0$, three clauses $(x\vee a_1 \vee a_2)$, $(y\vee a_2 \vee a_3)$ and $(z\vee a_1 \vee a_3)$ are created. Note that for any given assignment to $x$, $y$ and $z$ satisfying the equation mod $2$, it is possible to find an assignment of $a_1$, $a_2$ and $a_3$ that satisfies all clauses. Otherwise any assignment to $a_1$, $a_2$ and $a_3$ satisfies at most two clauses. The clause gadget as constructed in \cite{Karpinski2013b} is shown in \autoref{fig:G3v}.

\begin{figure}[htb]
 \centering
 \includegraphics[scale=1]{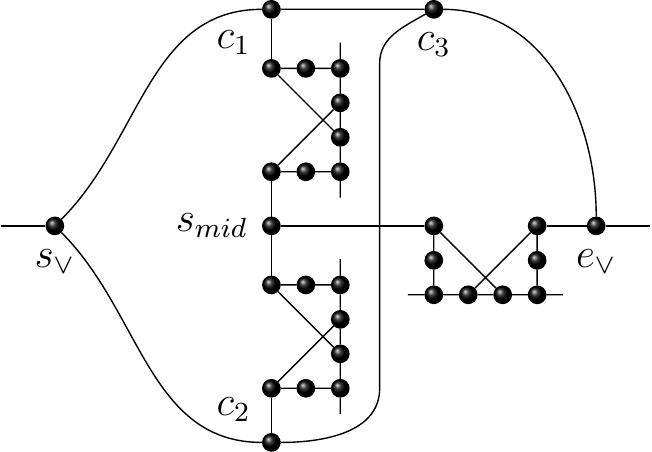}\quad\vrule\quad
 \includegraphics[scale=1]{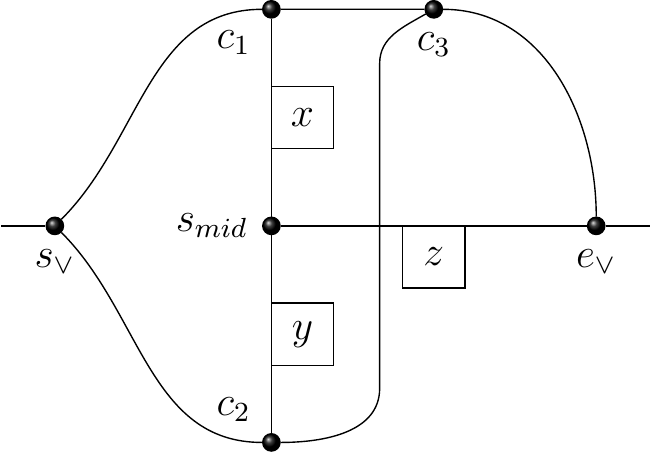}
 \caption{The gadget corresponding to clauses $(x\vee y\vee z)$ in detailed and modular view.}
 \label{fig:G3v}
\end{figure}

In order to complete the construction for an equation $x\oplus y\oplus z =0$, \cite{Karpinski2013b} introduced a gadget to simulate equations of the form $a_1^1\oplus a_1^2=0$. The construction is shown in \autoref{fig:Chain}.

\begin{figure}[htb]
 \centering
 \includegraphics[scale=1]{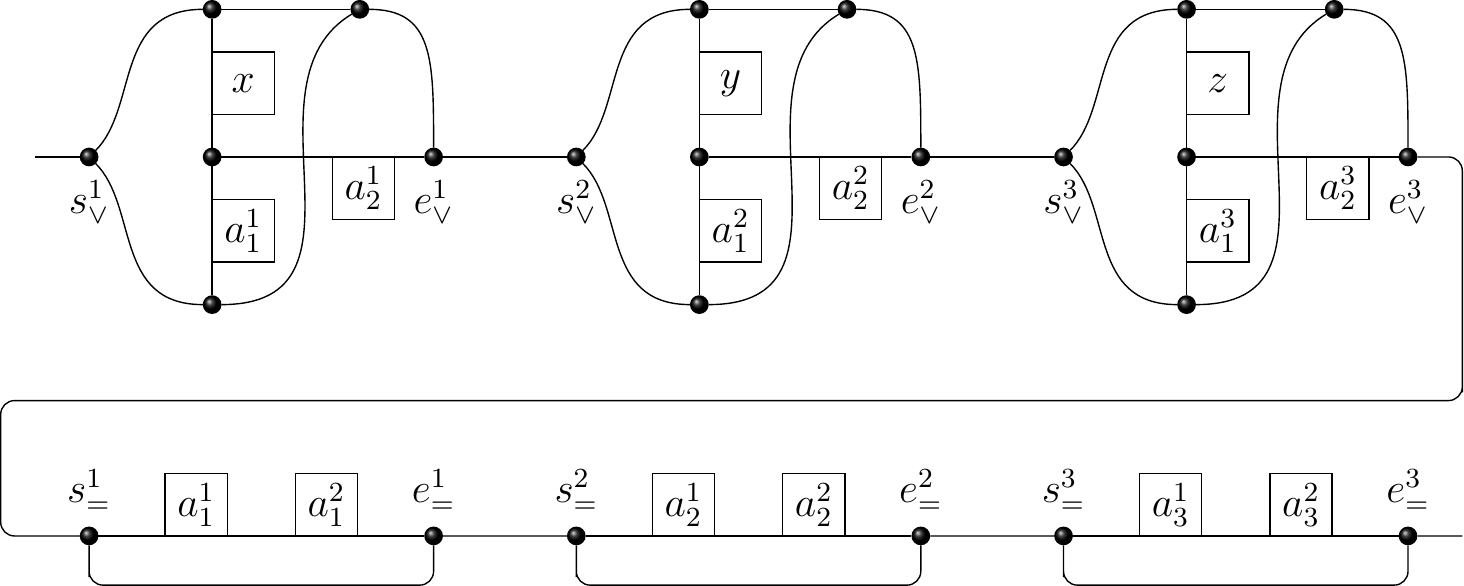}
 \caption{Modular view of the gadget simulating $x\oplus y\oplus z=0$.}
 \label{fig:Chain}
\end{figure}

This completes the modular description of the reduction of Hybrid instances $I_H$ to Subcubic $(1,2)$-TSP instances $G_{SC}^{12}$. For further details we refer to  \cite{Karpinski2013b}.

\paragraph{Subcubic $(1,2)$-TSP and Perfect Matchings}
Next, we modify the above instances $G_{SC}^{12}$ to permit a \emph{perfect matching} in order to prepare our reduction to $(1,2)$-TSP instances with underlying PLG.
First we observe that the subgraphs corresponding to parity gadgets (cf. \autoref{fig:Parity}) and the subgraphs corresponding to equations of the form $x\oplus y=0$ already contain a perfect matching.

Now we consider the clause gadget simulating $(x\vee y \vee z)$ as shown in \autoref{fig:G3v}.
It uses three isomorphic copies of a parity gadget and contains $6$ extra nodes, namely $s_\vee$, $s_{mid},$ $e_\vee$, $c_1$, $c_2$ and $c_3$ (cf. \autoref{fig:G3v}).
These extra nodes are of degree $3$. If we replace two such extra nodes $v$ by a separate copy of the graph $K_4$ and connect three of the four nodes 
to the neighbors of $v$ as shown in \autoref{fig:MatchingG3v}, then the resulting modified gadget has the additional property that it contains a perfect matching.

\begin{figure}[htb]
 \centering
 \includegraphics[scale=1]{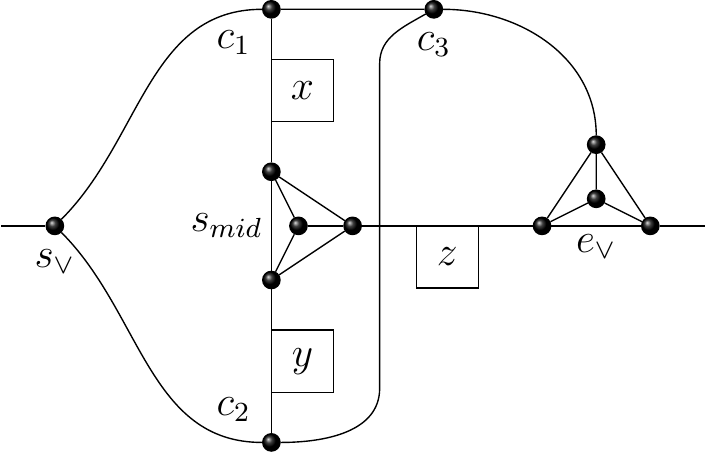}\raisebox{2.2cm}{$\quad\mapsto\quad$}
 \includegraphics[scale=1]{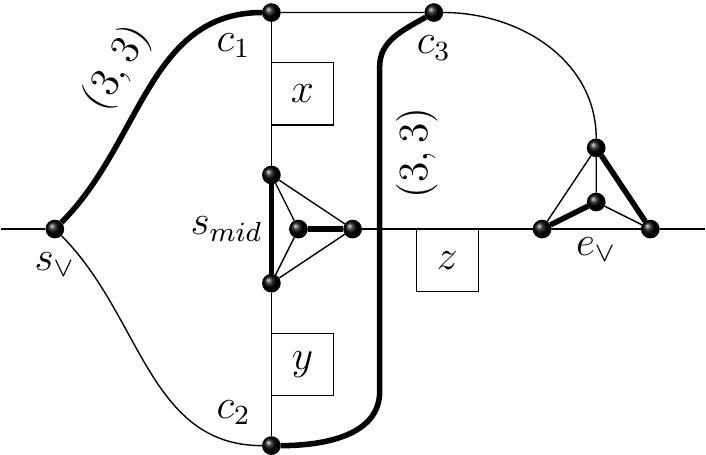}
 \caption{The modified gadget for clauses $(x\vee y\vee z)$ that contains a perfect matching.}
 \label{fig:MatchingG3v}
\end{figure}

If we let $\tilde{G}_{SC}^{12}$ the graph which is obtained from $G_{SC}^{12}$ by applying all these modifications, then we conclude that 
$\tilde{G}_{SC}^{12}$ contains a perfect matching which we denote as $M$. 

The modified instances $\tilde{G}_{SC}^{12}$ have $2\cdot (60v+2v\cdot 9)=156v$ nodes of degree $2$, $6\cdot (60v+2v\cdot 9)+2v\cdot (6\cdot 3+6)=516v$ nodes of degree $3$ 
and $3\cdot2\cdot \sharp G_v=3\cdot 6\cdot 2v=36v$ nodes of degree $4$.

Thus we have a mapping
\[\begin{array}{c@{}c@{}c@{}c@{}c}
\mbox{Hybrid} & \longrightarrow & \mbox{Subcubic $(1,2)$-TSP} & \longrightarrow & \mbox{$(1,2)$-TSP}\\[1.2ex]
I_H & \longmapsto & G_{SC}^{12} & \longmapsto & \tilde{G}_{SC}^{12}
\end{array}\]
where $\tilde{G}_{SC}^{12}$ is a graph of maximum degree $4$.
The resulting instance $\tilde{G}_{SC}^{12}$ has $708v$ vertices, namely $156v$ of degree $2$, $516$ of degree $3$ and $36v$ of degree $4$.
It is hard to decide if $\tau\geq 709v$ or $\tau\leq 708v$.

\paragraph{Reduction to Power Law $(1,2)$-TSP}
Now those instances have to be embedded into power law graphs. The instance $\tilde{G}_{SC}^{12}$ has a perfect matching $M$, which we want to use in order to embed
$\tilde{G}_{SC}^{12}$ into a power law graph. Thus we have to give a precise description of the node degrees of pairs of nodes which are connected by matching edges.
Let us make use of the following notation: We say an edge $e$ is of type $(a,b)$ if $e=\{u,v\}$, the degree of $u$ is $a$ and the degree of $v$ is $b$.
We observe that inside each parity gadget $P$, $M$ contains $2$ edges of type $(3,3)$ and $2$ edges of type $(2,3)$.

Moreover, $M$ contains for every $3$-equation $2\cdot 3+3$ edges of type $(3,3)$, $2\cdot 3$ edges of type $(3,4)$ and $2\cdot 3$ edges of type $(4,4)$.

Thus in total $M$ consists of the following edges:
\begin{itemize}
\item edges of type $(2,3)$: $2\cdot\sharp P=2\cdot (2v\cdot 9+60v)=2\cdot 78v=156v$
\item edges of type $(3,3)$: $2\cdot\sharp P+9\cdot\sharp\mbox{3-eq}=18v+2\cdot 78v=174v$
\item edges of type $(3,4)$: $6\cdot\sharp\mbox{3-eq}=12v$
\item edges of type $(4,4)$: $6\cdot\sharp\mbox{3-eq}=12v$ 
\end{itemize}  

We want to embed  $\tilde{G}_{SC}^{12}$ into an interval $[\delta ,\Delta ]$ of a power law distribution, where $\delta$ is a parameter of this construction.
In order to fit $\tilde{G}_{SC}^{12}$ into $[\delta ,\Delta ]$, we replace the edges of $M$ by multi-edges. Let $\mu\colon M\to {\mathbb N}$ such that 
for $e\in M$, $\mu (e)$ is the multiplicity of edge $e$. 

Let us first consider the special case when $\mu\colon M\to\{1\}$. In that case, we have to minimize $\alpha$ such as to satisfy the following inequalities: 
\[156v\leq\frac{e^{\alpha}}{2^{\beta}},\quad (156+174\cdot 2+12)v\leq\frac{e^{\alpha}}{3^{\beta}},\quad (12+2\cdot 12)v\leq\frac{e^{\alpha}}{4^{\beta}},\]
or equivalently
$e^{\alpha}\geq\max\{2^{\beta}\cdot 156v, 3^{\beta}\cdot 516v, 4^{\beta}\cdot 36v\}$.
We first observe that 
for $\beta\leq\frac{\log\left (\frac{516}{36}\right )}{\log (\frac{4}{3})}$, this is equivalent to $e^{\alpha}\geq 3^{\beta}\cdot 516v$.
Suppose now that we have chosen $\alpha$ such that $\zeta (\beta)e^{\alpha}=\zeta (\beta)\cdot 3^{\beta}\cdot 516v$.
The power law graph $G_{\alpha,\beta}$ is the union of $\mu (\tilde{G}_{SC}^{12})$, a matching $W_1$ on all but at most one degree $1$ nodes and a graph $W$ consisting of the remaining nodes of degree at least $2$ and at most one degree-$1$ node. 
The graph $W$ can be constructed in such a way that it contains a TSP path consisting of edges of cost one (cf. \cite{Gast2015}).
Let $|W|$ denote the number of nodes in the graph $W$.
Then we may assume every tour $\tau$ of the $(1,2)$-metric space defined by $G_{\alpha,\beta}$ to consist of a TSP path $\tau_{SC}$ of $\tilde{G}_{SC}^{12}$, a path in $W$ and a path of cost at most $\frac{3}{2}\cdot\e^{\alpha}-1$ connecting the degree-$1$ nodes in $W_1$ and three edges of cost $2$ connecting these paths. 
The cost of the TSP path in $W$ is asymptotically equal to $|W|-1$. Thus the cost of the tour $\tau$ satisfies
\[c(\tau) = c(\tau_{SC})+|W|+\frac{3}{2}e^{\alpha}+O(1).\] 
We have $|\tilde{G}_{SC}^{12}|=708v=708\cdot\frac{\zeta (\beta)e^{\alpha}}{\zeta (\beta)\cdot 3^{\beta}\cdot 516}$ and
\[|W|=\left (\zeta (\beta )-1-\frac{708}{516\cdot 3^{\beta}}\right )e^{\alpha}\]
Thus it is NP-hard to decide if the cost of an optimum tour $\tau^*$ in $G_{\alpha,\beta}$ satisfies
\[c(\tau^*)\leq 708\cdot\frac{e^{\alpha}}{3^{\beta}\cdot 516}+\left (\zeta (\beta )-1-\frac{708}{516\cdot 3^{\beta}}\right )e^{\alpha}+\frac{3}{2}\cdot e^{\alpha}+O(1)\]
or
\[c(\tau^*)\geq 709\cdot\frac{e^{\alpha}}{3^{\beta}\cdot 516}+\left (\zeta (\beta )+\frac{1}{2}-\frac{708}{516\cdot 3^{\beta}}\right )e^{\alpha}-O(1).\]

For $\beta\geq\frac{\log\left (\frac{516}{36}\right )}{\log (\frac{4}{3})}$, we obtain the requirement $e^{\alpha}\geq 4^{\beta}\cdot 36v$ and proceed in the same way.
We obtain the following result.
\begin{theorem}\label{thm:LowerOne}
For $1<\beta\leq\frac{\log\left (\frac{516}{36}\right )}{\log (\frac{4}{3})}$, the $(1,2)$-TSP on $(\alpha,\beta )$-Power Law Graphs is NP-hard to approximate within
any factor less than $\frac{3^{\beta}\cdot 516\cdot (\zeta (\beta)+\nicefrac{1}{2})+1}{3^{\beta}\cdot 516\cdot (\zeta (\beta)+\nicefrac{1}{2})}$.

For $\beta\geq\frac{\log\left (\frac{516}{36}\right )}{\log (\frac{4}{3})}$, the $(1,2)$-TSP on $(\alpha,\beta )$-Power Law Graphs is NP-hard to approximate within
any factor less than $\frac{4^{\beta}\cdot 36\cdot (\zeta (\beta)+\nicefrac{1}{2})+1}{4^{\beta}\cdot 36\cdot (\zeta (\beta)+\nicefrac{1}{2})}$.
\end{theorem}

Now we improve on this result and consider the general case when $\mu\colon M\to {\mathbb N}$. We consider an embedding which we call
\emph{Even-Degree Packing} for a given $\alpha$. We choose edge multiplicities $\mu (e)$ for $e=\{u,v\}$ such that when $d(u)+\mu(e)-1\leq d(v)+\mu(e)-1\leq d(u)+\mu (e)$,
then $d(u)+\mu(e)-1$ is even. This means that we partition the interval $[2,\Delta ]$ into subintervals of size $2$ of the form $[2i,2i+1]$ and construct the edge multiplicities
$\mu$ in such a way that we always have $\mu (e)\equiv d(u)-1 \mod 2$. Suppose the degrees which are occupied in this way are of the form $2i,1\leq i\leq x\Delta$.
Then the sum of multiplicities of edges which we can pack in this way is
\begin{align*}
\sum_{i=1}^{x\Delta}\left\lfloor\frac{e^{\alpha}}{(2i+1)^{\beta}}\right\rfloor
 & \geq    \sum_{i=1}^{x\Delta}\frac{e^{\alpha}}{(2i+1)^{\beta}} -x\Delta\\
 & \geq \int_{1}^{x\Delta+1}\frac{e^{\alpha}}{(2y+1)^{\beta}}dy -x\Delta
 = \frac{e^{\alpha}}{2(\beta -1)}\left (3^{1-\beta}-(2x\Delta+3)^{1-\beta} \right )-x\Delta
\end{align*}
with $x\leq\frac{1}{2}\cdot\frac{\Delta -1}{\Delta}$. Thus it suffices to choose $\alpha$ such that 
$\frac{e^{\alpha}}{2(\beta -1)}\cdot \frac{1}{3^{\beta -1}}-x\Delta\geq 354v$. For $x=\frac{1}{2}\cdot\frac{\Delta -1}{\Delta}$, this is equivalent to
\[\frac{e^{\alpha}}{3^{\beta -1}\cdot 2(\beta -1)}-\frac{1}{2}\cdot (e^{\alpha})^{1\slash\beta}+\frac{1}{2}\geq 354v\]
Now the cost of a tour $\tau$ is $c(\tau )=c(\tau_{SC})+|W|+\frac{3}{2}e^{\alpha}+3$, and we have 
$|\tilde{G}_{SC}^{12}|=708v=708\cdot\frac{e^{\alpha}}{3^{\beta -1}\cdot 2(\beta -1)\cdot 354}$. The size of the set $W$ is given as
\[|W|=(\zeta (\beta)-1)e^{\alpha}-|\tilde{G}_{SC}^{12}|=\left (\zeta (\beta)-1-\frac{708}{3^{\beta -1}\cdot 2(\beta -1)\cdot 354}\right )\]
Thus it is NP-hard to decide if
\[c(\tau^*)\leq 708\cdot \frac{e^{\alpha}}{3^{\beta -1}\cdot 2(\beta -1)\cdot 354}
                +\left (\zeta (\beta )-1-\frac{708}{3^{\beta -1}\cdot 2(\beta -1)\cdot 354}\right )e^{\alpha}+\frac{3}{2}e^{\alpha}+O(1)\]
or
\[c(\tau^*)\geq 709\cdot \frac{e^{\alpha}}{3^{\beta -1}\cdot 2(\beta -1)\cdot 354}
                +\left (\zeta (\beta )+\frac{1}{2}-\frac{708}{3^{\beta -1}\cdot 2(\beta -1)\cdot 354}\right )e^{\alpha}-O(1)\]
                
We obtain the following result.
\begin{theorem}\label{thm:LowerTwo}
The (1,2)-TSP on $(\alpha,\beta )$-Power Law Graphs is NP-hard to approximate within
any factor less than $\frac{(\zeta (\beta)+\nicefrac{1}{2})\cdot 3^{\beta -1}\cdot 2(\beta -1)\cdot 354+1}{(\zeta (\beta)+\nicefrac{1}{2})\cdot 3^{\beta -1}\cdot 2(\beta -1)\cdot 354}$.
\end{theorem}

In \autoref{fig:PlotLowerBound} we plot and compare the resulting lower bounds of \autoref{thm:LowerOne} and \autoref{thm:LowerTwo} for the $(1,2)$-TSP with underlying $(\alpha,\beta )$-PLG.

\begin{figure}[htb]
 \centering
  \input{PlotLowerBound}
  \caption{Plot of the lower bounds due to \autoref{thm:LowerOne} and \autoref{thm:LowerTwo}.}
  \label{fig:PlotLowerBound}
\end{figure}
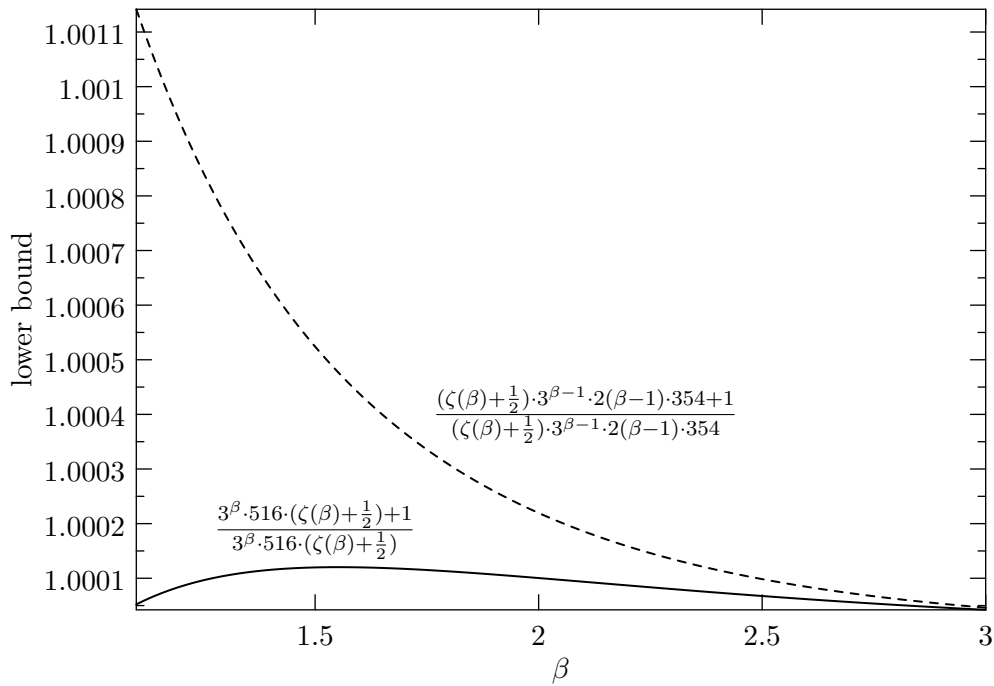



%% file: PlotLowerBound.tex
\begin{asy}
 defaultpen(fontsize(11));
 import graph;
 import gsl;
     size(13cm,9cm,IgnoreAspect);

     pen dashed=linetype(new real[] {4,4});
     pen p=linewidth(.75);

     //real f(real x) {return (3^(x)*516*(zeta(x)+1)+1)/(3^(x)*516*(zeta(x)+1));}
     real f(real x) {return (3^(x)*516*(zeta(x)+1/2)+1)/(3^(x)*516*(zeta(x)+1/2));}
     pair F(real x) {return (x,f(x));}
     //real g(real x) {return ((zeta(x)+1)*3^(x-1)*2(x-1)*354+1)/((zeta(x)+1)*3^(x-1)*2(x-1)*354);}
     real g(real x) {return ((zeta(x)+1/2)*3^(x-1)*2(x-1)*354+1)/((zeta(x)+1/2)*3^(x-1)*2(x-1)*354);}
     pair G(real x) {return (x,g(x));}
     //real h(real x) {return 535/534;}
  
     draw("$$",graph(f,1.1,3),p);
     draw("$$",graph(g,1.1,3),dashed+p);
     //draw("$$",graph(h,1.1,3),dotted+p);
     //draw("$$",graph(i,1.1,3),dotted+p);
     //draw("$\frac{535}{534}$",graph(h,1.1,3),dotted+p);
     
 //    ylimits(1.968,2);
     xaxis("$\beta$",BottomTop,LeftTicks);
     yaxis("lower bound",LeftRight,RightTicks);
     label("$\frac{3^{\beta}\cdot 516\cdot (\zeta (\beta)+\frac{1}{2})+1}{3^{\beta}\cdot 516\cdot (\zeta (\beta)+\frac{1}{2})}$",F(1.5),N);
     label("$\frac{(\zeta (\beta)+\frac{1}{2})\cdot 3^{\beta -1}\cdot 2(\beta -1)\cdot 354+1}{(\zeta (\beta)+\frac{1}{2})\cdot 3^{\beta -1}\cdot 2(\beta -1)\cdot 354}$",G(1.75),NE);

     //attach(legend(),(point(S).x,truepoint(S).y),10S,UnFill);
  \end{asy}